\documentclass[a4paper,10pt]{article}
\usepackage{graphicx}
\usepackage{color}\usepackage{amsmath, amsthm, slashed}
\usepackage{amssymb}
\usepackage{rotating}
\usepackage{epsfig}
\usepackage{verbatim}
\usepackage{rotating}
\usepackage{graphicx}

\newtheorem{theorem}{Theorem}[section]

\newtheorem{corollary}{Corollary}[section]
\newtheorem{proposition}{Proposition}[section]
\newtheorem{lemma}{Lemma}[section]
\newtheorem{remark}{Remark}[section]

\newcommand{\nabb}{\mbox{$\nabla \mkern-13mu /$\,}}

\title{Quasimodes and a Lower Bound on the Uniform Energy Decay Rate for Kerr-AdS Spacetimes}

\author{Gustav Holzegel\footnote{Imperial College London, Department of Mathematics, South Kensington Campus, London SW7 2AZ, United Kingdom.}
\, and Jacques Smulevici\footnote{Laboratoire de Math\'ematiques, Universit\'e Paris-Sud 11, b\^at. 425, 91405 Orsay, France.}
}

\begin{document}
\maketitle

\begin{abstract}
We construct quasimodes for the Klein-Gordon equation on the black hole exterior of Kerr-Anti-de Sitter (Kerr-AdS) spacetimes. Such quasimodes are associated with time-periodic approximate solutions of the Klein Gordon equation
 and provide natural candidates to probe the decay of solutions on these backgrounds.
They are constructed as the solutions of a semi-classical non-linear 
eigenvalue problem arising after separation of variables, with the (inverse of the) angular momentum playing the role of the semi-classical parameter. Our construction results in \emph{exponentially small} errors in the semi-classical parameter. This implies that general solutions to the Klein Gordon equation on Kerr-AdS cannot decay faster than logarithmically. The latter result completes previous work by the authors, where a logarithmic decay rate was established as an upper bound.
\end{abstract}

 \tableofcontents

\section{Introduction}
There is currently a lot of mathematical activity concerning the analysis of waves on the exterior of black hole backgrounds.\footnote{See \cite{mihalisnotes} for an introduction and a review of recent results.} The main motivation is the black hole stability problem, i.e.~the conjectured non-linear asymptotic stability of the two-parameter family of asymptotically-flat Kerr spacetimes $\left(\mathcal{M},g_{M,a}\right)$, the latter being stationary solutions of the vacuum Einstein equations $Ric\left[g\right]=0$. From the point of view of non-linear partial differential equations, the analysis of linear scalar waves on black holes is a prerequisite to successfully understand the non-linear hyperbolic Einstein equations in a neighborhood of the Kerr family. 

A major insight that crystallized in the last decade  \cite{DafRod, DafRod2, DafRodsmalla, DafRodlargea, Toha1, Toha2, AndBlue, Dyatlov1, Dyatlov2} is that the fundamental geometric obstacles to the decay of waves, namely superradiance and trapped null-geodesics, can be overcome by exploiting the normal hyperbolicity of the trapping, the redshift effect near the event horizon and the natural dispersion of waves in asymptotically flat (and asymptotically de Sitter) spacetimes. In particular, polynomial decay rates have been established for solutions to the scalar wave equation on the exterior of any member of the sub-extremal Kerr family of spacetimes \cite{DafRodlargea}.


Changing the black hole geometry can have dramatic effects on the behavior of linear waves through the subtle interplay of the redshift, the superradiance and the trapping. Aretakis \cite{Aretakis, Aretakis2} showed that for \emph{extremal} black holes (whose vanishing surface gravity leads to a degeneration of the redshift effect) the transversal derivatives of general solutions to the wave equation will grow along the event horizon. In \cite{ShlapRoth}, it is proven that for the \emph{massive} wave equation on a sub-extremal ($|a|<M$) Kerr spacetime, exponentially growing solutions can be constructed on the exterior, exploiting an amplification of the superradiance caused by the confining properties of the mass term. 

 In this paper, we shall be interested in a black hole geometry for which a strong trapping phenomenon leads to a very slow (only logarithmic) decay of waves. More precisely, we will study the behavior of solutions to the massive wave equation
 \begin{equation} \label{mwe}
\left(\square_g +\frac{\alpha}{l^2}\right) \psi=0
\end{equation} 
in the exterior of asymptotically Anti-de-Sitter (AdS) black holes with spacetime metric $g$.
 
Due to their AdS asymptotics, these spacetimes are not globally hyperbolic. Nonetheless, the equation (\ref{mwe}) is well posed in suitably weighted Sobolev spaces, denoted here by $H^k_{AdS}$, provided $\alpha$ satisfies the Breitenlohner-Friedmann bound $\alpha<9/4$. See \cite{Holzegelwp} as well as \cite{Vasy2, Bachelot, Bachelot2, Ishibashi} and \cite{cw:mweads} for a complete treatment of general boundary conditions. 

The \emph{global} properties of solutions to (\ref{mwe}) on the exterior of non-superradiant\footnote{This means that the parameters of the black hole satisfy $r_+^2> |a|l$. See Remark \ref{rem:uh} and Section \ref{sec:prelim}.} Kerr-AdS black holes were studied in \cite{HolzegelAdS, gs:dpkads, CWgh}. In particular, boundedness was obtained in \cite{HolzegelAdS, CWgh} and logarithmic decay in time for general $H^2_{AdS}$ solutions in \cite{gs:dpkads}. We summarize these results in the following theorem. We refer to Section \ref{se:fmb} for the precice definitions of the Kerr-AdS spacetimes, the area-radius $r_+$ of the event horizon and the $\Sigma_{t^\star}$--foliation and to Section \ref{sec:norms} for the definitions of the norms and energies used in the statement below. At this point we only remark that $e_1\left[\psi\right]$ is an energy density involving all first derivatives of $\psi$ while $e_2\left[\psi\right]$ and $\tilde{e}_2\left[\psi\right]$ involve all second derivatives (with appropriate weights):

\begin{theorem} \label{theo:pre}
Let $(g,\mathcal{R})$ denote the black hole exterior of a Kerr-AdS spacetime with mass $M>0$, angular momentum per unit mass $a$ and cosmological constant $\Lambda=-\frac{3}{l^2}$. Assume that the parameters satisfy $\alpha < \frac{9}{4}$, $|a| < l$. Fix a spacelike slice $\Sigma_{t^\star_0}$ intersecting $\mathcal{H}^+$. Then the following is true.
\begin{enumerate}
\item Equation (\ref{mwe}) is well-posed in $CH^k_{AdS}$ on $(g,\mathcal{R})$ for any $k\geq 2$ for initial data prescribed on $\Sigma_{t^\star_0}$. See \cite{Holzegelwp}.
\item \label{labbounded} The solutions of (\ref{mwe}) arising from data prescribed on $\Sigma_{t^\star_0}$ remain uniformly bounded on the black hole exterior provided $r_+^2 > |a|l$ holds. In particular,
\begin{equation} \label{bdness}
\int_{\Sigma_{t^\star}} e_1\left[\psi\right] \left(t^\star\right) \ r^2 \sin \theta dr d\theta d\phi \lesssim \int_{\Sigma_{t^\star_0}} e_1\left[\psi\right] \left(t^\star_0\right) \ r^2 \sin \theta dr d\theta d\phi \, .
\end{equation}
Analogous statements hold for all higher $H^k_{AdS}$-norms. In particular,
\begin{align} \label{bndhigher}
\int_{\Sigma_{t^\star}} e_2\left[\psi\right] \left(t^\star\right) \ r^2 \sin \theta dr d\theta d\phi \lesssim \int_{\Sigma_{t^\star_0}} e_2\left[\psi\right] \left(t^\star_0\right) \ r^2 \sin \theta dr d\theta d\phi \, .
\end{align}
and the same statement for $\tilde{e}_2\left[\psi\right]\left(t^\star\right)$. See  \cite{HolzegelAdS, CWgh}.\footnote{The aforementioned papers (as well as \cite{Holzegelwp}) are only concerned with the $\tilde{e}_2\left[\psi\right]$-energy. It is remarked that by commutation with angular momentum operators one can prove boundedness for the ${e}_2\left[\psi\right]$-energy (which differs from the  $\tilde{e}_2\left[\psi\right]$-energy through the weights of the angular derivatives). For completeness, we provide an explicit proof of this statement in the appendix of this paper.}
\item \label{labdecay} If the parameters satisfy $r_+^2 > |a|l$, the solutions of (\ref{mwe}) satisfy the following global decay estimate:
\begin{align} \label{decay}
\int_{\Sigma_{t^\star}} e_1\left[\psi\right] \left(t^\star\right) \ r^2 \sin \theta dr d\theta d\phi \lesssim   \frac{1}{\left[\log( 2+ t^\star)\right]^2}  \int_{\Sigma_{t^\star_0}} e_2\left[\psi\right] \left(t^\star_0\right) \ r^2 \sin \theta dr d\theta d\phi
\end{align}
for all $t^\star \ge t^\star_0 > 0$. See \cite{gs:dpkads}.
\end{enumerate}
\end{theorem}
\begin{remark}
The constant implicit in the symbol ``$\lesssim$" appearing in \ref{labbounded}. and \ref{labdecay}. depends only on the fixed parameters $M$, $\ell$, $a$ and $\alpha$.
\end{remark}

\begin{remark} \label{rem:uh}
The condition $r_+^2 > |a| l$ on the parameters in \ref{labbounded}.~and \ref{labdecay}. guarantees the existence of a globally causal Killing vectorfield on the black hole exterior, the Hawking-Reall vectorfield \cite{HawkingReall}, which explains why such black holes are sometimes referred to as ``non-superradiant". If one restricts to \underline{axisymmetric} solutions of (\ref{mwe}), this condition can be dropped for both \ref{labbounded}.~and \ref{labdecay}. in Theorem \ref{theo:pre}.
\end{remark}
\begin{remark}
Note that the boundedness statement (\ref{bdness}) does not lose derivatives, while the decay statement (\ref{decay}) does. This is the familiar loss of derivatives caused by the existence of trapped null-geodesics \cite{jr:swle}.
\end{remark}

This logarithmic decay rate (\ref{decay}) was conjectured to be sharp in  \cite{gs:dpkads} in view of the discovery of a new stable trapping phenomenon, itself a consequence of the coupling between the lack of dispersion at the asymptotic end and the usual (unstable) trapping on black hole exteriors. See again \cite{gs:dpkads}. 


\subsection{The main results}
In this paper, we shall prove that the logarithmic decay estimate of Theorem \ref{theo:pre} is indeed sharp. Recall that for the obstacle problem, it is classical \cite{jr:swle} that lower bounds on the rate of energy decay can be obtained from the construction of approximate eigenfunctions, also called \emph{quasimodes}, of the associated elliptic operator, obtained by formally taking the Fourier transform in time of the wave operator. Our main theorem establishes the existence of such quasimodes for Kerr-AdS with exponentially small errors.\footnote{Note that in order to deduce the sharpness of the logarithmic decay rate from the quasimodes, polynomial errors would a priori not be sufficient.} 

The statement of the following theorem will involve the quantity $r_{max} \in \left(r_{+},\frac{3M}{1-\frac{a^2}{l^2}}\right]$, which is determined in Lemma \ref{lem:popro} as the location of the unique maximum of a simple radial function. For Schwarzschild-AdS, $r_{max}=3M$ will be the location of the well-known photon sphere.

\begin{theorem}[Quasimodes for Kerr-AdS]\label{th:main}
Let $(g,\mathcal{R})$ denote the black hole exterior of a Kerr-AdS spacetime, with mass $M>0$, angular momentum per unit mass $a$ and cosmological constant $\Lambda=-\frac{3}{l^2}$. Assume that the parameters satisfy $\alpha < \frac{9}{4}$, $|a| < l$. Let $(t,r,\theta,\varphi)$ denote standard Boyer-Lindquist coordinates on $\mathcal{R}$. Then, for $\delta > 0$ sufficiently small (depending only on the parameters $l$, $M$, $a$, $\alpha$), there exists a family of non-zero functions $\psi_\ell\in H^k_{AdS}$ for any $k\geq 0$ such that
\begin{enumerate}
\item $\psi_{\ell}\left(t,r,\theta,\varphi\right)=e^{i \omega_\ell t}\varphi_\ell(r,\theta)$
(axisymmetric and time-periodic), \label{prop1}
\item $0<c< \frac{\omega^2_\ell}{\ell\left(\ell+1\right)}<C$, for constants $c$ and $C$ independent of $\ell$ \\(uniform bounds on the frequencies), \label{prop2}
\item for all $t^\star \ge t^\star_0$, for all $k \ge 0$, $||\left(\square_g+\frac{\alpha}{l^2}\right)\psi_{\ell}||_{H^k_{Ads}(\Sigma_{t^\star})} \le C_k e^{-C_k \ell }||\psi_{\ell}||_{H^0_{AdS}(\Sigma_{t^\star_0})}$, for some $C_k > 0$ independent of $\ell$ (approximate solutions to the wave equation), \label{prop3}
\item the support of $F_\ell:=\left(\square_g+\frac{\alpha}{l^2}\right)\psi_{\ell}$ is contained in $\{ r_{max}\le r \le r_{max}+\delta \}$ (spatial localization of the error), \label{prop4}
\item the support of $\varphi_\ell(r,\theta)$ is contained in $\{r \geq r_{max} \}$ \\(spatial localization of the solution).  \label{prop5}
\end{enumerate}
\end{theorem}
Note that the $\psi_{\ell}$ have constant $H^k_{AdS}$-norms and hence exhibit no decay. On the other hand, a standard application of Duhamel's formula shows that the $\psi_\ell$ are good approximations to the solution of $\left(\square_g+\frac{\alpha}{l^2}\right) \psi=0$ arising from the data induced by $\psi_{\ell}$, at least up to a time $t \sim e^{C_k \ell}$. 

\begin{corollary}\label{cor}
Let $(\mathcal{R},g)$ denote the black hole exterior of a Kerr-AdS spacetime as in Theorem \ref{th:main}. Denote by $SCH^2_{AdS}$ the set of $CH^2_{AdS}$ solutions to \eqref{mwe} with $\alpha<\frac{9}{4}$. Let $t^\star_0\geq 0$ be fixed and define for any non-zero $\psi$ and $t^\star \geq 0$
\[
Q \left[\psi\right] \left(t^\star\right) := \log(2+t^\star) \left[\frac{\int_{\Sigma_{t^\star} \cap \{r\geq r_{max}\}} e_1\left[\psi\right] \left(t^\star\right) \ r^2 \sin \theta dr d\theta d\phi}{\int_{\Sigma_{t^\star_0}} e_2\left[\psi\right] \left(t^\star_0\right) \ r^2 \sin \theta dr d\theta d\phi}\right]^\frac{1}{2} \, .
\] 
Then there exists a universal (depending only on $M$, $\alpha$, $|a|$ and $l$) constant $C > 0$ such that
\[
\limsup_{t^\star \rightarrow + \infty} \sup_{ \psi \in SCH^{2}_{AdS}, \psi\neq 0 } Q \left[\psi\right] \left(t^\star\right) > C > 0.
\]
\end{corollary}

\begin{remark}
Corollary \ref{cor} implies that the (semi-)local energy in $\{r \ge r_{max}\}$ cannot decay universally faster than $\left(\log t^\star\right)^{-2}$, unless one loses more derivatives. 
\end{remark}

\begin{remark}
We emphasize that no smallness assumption on the angular momentum $a$ is needed apart from the condition $|a| < l$ which ensures that the metric is a regular black hole metric.
\end{remark}

\begin{remark} \label{rem:surad}
Note that the $\psi_\ell$ constructed in Theorem \ref{th:main} are axisymmetric, while the decay estimate of Theorem \ref{theo:pre} holds for general solutions (provided $r_+^2 > |a|l$ holds in the non-axisymmetric case). Since we are concerned here with a \underline{lower} bound on the uniform decay rate, an analysis within axisymmetry is sufficient. This allows us to drop the non-superradiant condition $r_+^2 > |a|l$ in the analysis. On the other hand, for (sufficiently) superradiant black holes $r_+^2 <|a|l$ one can adapt the proof of \cite{ShlapRoth} to construct exponentially growing solutions. Hence in this case, the quasimodes we construct are not the ``worst" solutions on these backgrounds.
\end{remark}

\begin{remark}
Note that the stable trapping occurs only in the region $r \ge r_{\max}$ and is associated to certain frequencies. As a consequence, stronger local energy decay in $ r \le r_{max}$ or for some frequency projections of solutions is a priori compatible with the results of this paper.
\end{remark}


\subsection{Related works and discussion}
\subsubsection{Non-linear analysis on asymptotically AdS spacetimes}
In \cite{gs:lwp,gs:stab}, the non-linear spherically-symmetric Einstein-Klein-Gordon system for asymptotically AdS initial data was studied, and in particular, the asymptotic stability of Schwarzschild-AdS was proven within this model. 

For a discussion connecting the logarithmic decay to the non-linear stability or instability of asymptotically AdS black holes, we refer to Section 1.4 of \cite{gs:dpkads}. We also mention the recent heuristic analysis of \cite{dhms:nsaads} drawing attention to a potential stability mechanism caused by the lack of exact non-linear resonances in this setting. For AdS itself, instability was conjectured in \cite{DafHolnotes, AndAdS} and \cite{Newtontalk}. More recently, both numerical and additional heuristic evidence has been presented \cite{BizonAdS}.
Finally, let us note that asymptotically AdS solutions to the Einstein equations have been constructed in \cite{FriedrichAdS}.

\subsubsection{Quasi-normal modes of the asymptotically AdS black holes}
Quasi-normal modes, also called resonances, are complex frequencies generalizing the well-known normal modes to systems which dissipate energy. There is a strong connection between quasimodes and resonances \cite{tz:fqr}. One way to mathematically define them is as poles of the meromophic continuation of a truncated resolvent. In the case of asymptotically de-Sitter black holes, this theory has been very successfully developed, see \cite{Haefner, Dyatlov1, Dyatlov2}. In a recent paper \cite{Gannot}, Gannot has established, in the case of Schwarzschild-AdS, the existence of a sequence of quasi-normal modes (based on an independent construction of quasimodes), indexed by angular momentum $\ell$, and with imaginary parts of size $\mathcal{O}\left(\exp(-\ell/C)\right)$. In particular, his construction confirms the numerical results of \cite{Festuccia} and provides an independent proof of Theorem \ref{th:main} and Corollary \ref{cor} albeit restricted to the Schwarzschild-AdS case.
While we restrict ourselves here to the construction of quasimodes, we strongly believe that our results can be used as a basis for the construction of resonances in the Kerr-AdS case for the whole range of parameters satisfying $|a|< l$ and $r_+^2 > |a|l$.

\subsubsection{Universal minimal decay rates of waves outside stationary black holes}
In the context of the obstacle problem in Minkowski space, a celebrated result of Burq \cite{nb:del} establishes a logarithmic decay rate for the local energy of waves, independently of the geometry of the obstacle causing the trapping. For waves outside black holes, in view of the results of \cite{gs:dpkads}, a natural conjecture is: Given any black hole exterior $(\mathcal{R},g)$ of a stationary spacetime, a logarithmic decay of energy similar to that of \cite{gs:dpkads} will hold, provided a uniform boundedness statement is true for solutions to the wave equation $\square_g \psi=0$ on $\mathcal{R}$.

\subsection{Outline and overview of the proof}
Section \ref{sec:prelim} introduces the family of Kerr-AdS spacetimes as well as the norms required to state our estimates. In Section \ref{se:sere}, we exploit the classical fact that the wave equation separates on Kerr-AdS. An ingredient which considerably simplifies our analysis here is the important observation that we can restrict ourselves to \emph{axisymmetric} solutions. With axisymmetry, the separation of variables leads to relatively simple, one-dimensional, second order ordinary differential equations for the radial functions. In the case of Schwarzschild, they are roughly of the form of the semi-classical problem
\begin{equation} \label{semiclass}
-u^{\prime \prime}_\ell \frac{1}{\ell \left(\ell+1\right)} + V_{\sigma} u_\ell = \frac{\omega^2}{\ell \left(\ell +1\right)} u_\ell
\end{equation}
for a potential $V_\sigma\left(r\right)$, whose general form is depicted below.\footnote{For the purpose of this exposition, we neglect terms of lower order in $\frac{1}{\ell\left(\ell+1\right)}$ in the potential, as well as the mass-term. The latter is actually unbounded and needs to be absorbed with a Hardy inequality. We suppress such technical difficulties in the present discussion.} In Section \ref{sec:potential}, we shall describe in detail the analytic properties of the potentials appearing in these equations.
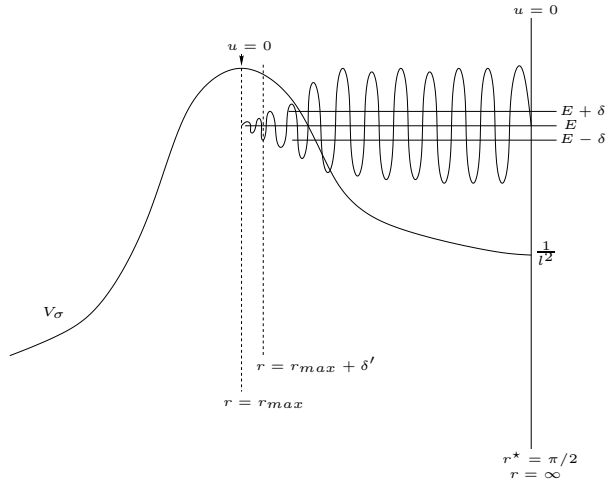
\[
\begin{picture}(0,0)%
\includegraphics{potential.pstex}%
\end{picture}%
\setlength{\unitlength}{1184sp}%
\begingroup\makeatletter\ifx\SetFigFont\undefined%
\gdef\SetFigFont#1#2#3#4#5{%
  \reset@font\fontsize{#1}{#2pt}%
  \fontfamily{#3}\fontseries{#4}\fontshape{#5}%
  \selectfont}%
\fi\endgroup%
\begin{picture}(11502,9994)(889,-9311)
\put(11776,-4561){\makebox(0,0)[lb]{\smash{{\SetFigFont{5}{6.0}{\rmdefault}{\mddefault}{\updefault}{\color[rgb]{0,0,0}$\frac{1}{l^2}$}%
}}}}
\put(12376,-1936){\makebox(0,0)[lb]{\smash{{\SetFigFont{5}{6.0}{\rmdefault}{\mddefault}{\updefault}{\color[rgb]{0,0,0}$E$}%
}}}}
\put(11101,-8911){\makebox(0,0)[lb]{\smash{{\SetFigFont{5}{6.0}{\rmdefault}{\mddefault}{\updefault}{\color[rgb]{0,0,0}$r^\star=\pi/2$}%
}}}}
\put(11251,-9211){\makebox(0,0)[lb]{\smash{{\SetFigFont{5}{6.0}{\rmdefault}{\mddefault}{\updefault}{\color[rgb]{0,0,0}$r=\infty$}%
}}}}
\put(5326,-7711){\makebox(0,0)[lb]{\smash{{\SetFigFont{5}{6.0}{\rmdefault}{\mddefault}{\updefault}{\color[rgb]{0,0,0}$r=r_{max}$}%
}}}}
\put(12301,-1636){\makebox(0,0)[lb]{\smash{{\SetFigFont{5}{6.0}{\rmdefault}{\mddefault}{\updefault}{\color[rgb]{0,0,0}$E+\delta$}%
}}}}
\put(5401,-286){\makebox(0,0)[lb]{\smash{{\SetFigFont{5}{6.0}{\rmdefault}{\mddefault}{\updefault}{\color[rgb]{0,0,0}$u=0$}%
}}}}
\put(11326,464){\makebox(0,0)[lb]{\smash{{\SetFigFont{5}{6.0}{\rmdefault}{\mddefault}{\updefault}{\color[rgb]{0,0,0}$u=0$}%
}}}}
\put(6001,-6961){\makebox(0,0)[lb]{\smash{{\SetFigFont{5}{6.0}{\rmdefault}{\mddefault}{\updefault}{\color[rgb]{0,0,0}$r=r_{max}+\delta^\prime$}%
}}}}
\put(12301,-2236){\makebox(0,0)[lb]{\smash{{\SetFigFont{5}{6.0}{\rmdefault}{\mddefault}{\updefault}{\color[rgb]{0,0,0}$E-\delta$}%
}}}}
\put(1576,-5836){\makebox(0,0)[lb]{\smash{{\SetFigFont{5}{6.0}{\rmdefault}{\mddefault}{\updefault}{\color[rgb]{0,0,0}$V_{\sigma}$}%
}}}}
\end{picture}%

\]
To construct the quasimodes, we first construct eigenfunctions for the problem (\ref{semiclass}) with Dirichlet conditions $u=0$ imposed on $u$ at $r=r_{max}$ and $r\rightarrow \infty$ (Section \ref{se:te}). In particular, we prove a version of Weyl's law, ensuring that for any energy between $\frac{1}{l^2}<E<V_{max}=V_\sigma\left(r_{max}\right)$ we can find (lots of) eigenvalues $\frac{\omega_\ell^2}{\ell\left(\ell+1\right)}$ of (\ref{semiclass}) in a strip $\left[E-\delta,E+\delta\right]$. In the Kerr case, the eigenvalue problem (\ref{semiclass}) turns into a problem of the form\footnote{The $\mu_{\ell}\left(a^2\omega^2\right)$ generalize the familiar spherical eigenvalues $\ell\left(\ell+1\right)$ of the Schwarzschild case to Kerr. See Section \ref{se:sphhar}.}
 \begin{equation} \label{semiclass2}
-u^{\prime \prime}_\ell \frac{1}{\mu_\ell\left(a^2\omega^2\right)} + V_{\sigma} u_\ell = \frac{\omega^2}{\mu_\ell\left(a^2\omega^2\right)} u_\ell \, ,
\end{equation}
which is \emph{non-linear} in $\omega^2$. An application of the implicit function theorem together with global estimates on the behavior of the eigenvalues still allows us to conclude the existence of eigenfunctions $u_\ell$ of (\ref{semiclass2}) with corresponding eigenvalues in the range $\left(\frac{1}{l^2},E+\delta\right]$. \emph{These estimates, together with the analysis of the potential in the Kerr-AdS case, constitute the core of our paper.}

In Section \ref{se:ae}, we recall the so-called Agmon estimates, their proof being included to make the paper self-contained. These estimates quantify that the solutions $u_\ell$ constructed from the above eigenvalue problem decay exponentially in $\ell$ in a region $\left[r_{max},r_{max}+\delta^\prime\right]$. 

 In Section \ref{se:cq}, the quasimodes are constructed by cutting off the solution $u_\ell$ of the eigenvalue problem in $\left[r_{max},r_{max}+\delta^\prime\right]$ so that it vanishes with all derivatives at $r_{max}$ and then continuing it to be identically zero in $\left[r_{hoz},r_{max}\right]$. The function $\phi_\ell$ thus constructed will be defined on $\left[r_{hoz},\infty\right)$ and the corresponding wave function $\psi_\ell = e^{i\omega_\ell} \phi_\ell$ will satisfy the wave equation everywhere except in the small strip $\left[r_{max},r_{max}+\delta^\prime\right]$, where the error is exponentially small by the Agmon estimate.
 
In the last section, we prove Corollary \ref{cor} using the Duhamel formula. Finally, Appendix \ref{app} contains a proof of boundedness for the second energy used in this paper. This boundedness statement differs from that obtained in \cite{HolzegelAdS, CWgh} in that it allows for stronger radial weights near infinity for the angular derivatives.

\subsection{Acknowledgments}
Theorem \ref{th:main} and Corollary \ref{cor} were announced for the Schwarzschild case at the workshop ``Mathematical Aspects of General Relativity" in Oberwolfach, Germany, in August 2012 \cite{Wolf}. We thank Matthieu L\'eautaud for showing us the Agmon estimates and Maciej Zworski for pointing out the work of Gannot \cite{Gannot}. G.H.~acknowledges support through NSF grant DMS-1161607 and thanks the Department of Mathematics at Orsay for its hospitality.

\section{Preliminaries} \label{sec:prelim}
\subsection{The Kerr-AdS family of spacetimes} \label{se:sads}
We recall here some basic facts about the family of Kerr-AdS spacetimes required in the paper. We refer the reader to the detailed discussion in our \cite{gs:dpkads}.

\subsubsection{The fixed manifold with boundary $\mathcal{R}$} \label{se:fmb}
Let $\mathcal{R}$ denote the manifold with boundary
\[
\mathcal{R}=[0,\infty) \times \mathbb{R} \times \mathbb{S}^2 \, .
\]
We define standard coordinates $y^*$ for $[0,\infty)$, $t^\star$ for $\mathbb{R}$ and $(\theta,\phi)$ for $\mathbb{S}^2$. This defines a coordinate system on $\mathcal{R}$, which is global up to the well-known degeneration of the spherical coordinates.  We define the event horizon $\mathcal{H}^+$ to be the boundary of $\mathcal{R}$:
$$
\mathcal{H}^+=\partial \mathcal{R}=\{ y^*=0 \}.
$$ 
The manifold $\mathcal{R}$ will coincide with the domain of outer communication of the black hole spacetimes including the future event horizon $\mathcal{H}^+$.

\subsubsection{The parameter space and the radial function}
A Kerr-AdS spacetime is characterized by its mass $M>0$, its angular momentum per unit mass $a$ and the value of the cosmological constant $\Lambda < 0$. For convenience, we shall use mostly $l= \sqrt \frac{3}{| \Lambda |}$, instead of $\Lambda$. For the spacetime to be regular, we will require $|a| < l$.
Let us thus fix $M,l > 0$ and $|a| < l$. We then define $r_+(M,l,a)> 0$ as the unique real solution of $\Delta_-(x)=0$ where

\begin{equation}\label{def:deltam}
\Delta_-(x)=\left( 1+\frac{x^2}{l^2}\right)(x^2+a^2)-2Mx \, .
\end{equation}

We now define a function $r$ on $\mathcal{R}$ as follows. As a function of $(t^\star,y^\star,\varphi,\theta)$, $r$ only depends on $y^\star$ and is a diffeomorphism from $[0,\infty)$ to $[r_+,\infty)$. The collection $(t^\star,r(y^\star),\varphi,\theta)$ then form a coordinate system on $\mathcal{R}$, global up to the degeneration of the spherical coordinates. 
Moreover, the horizon $\mathcal{H}^+$ coincides with $\{ r =r_+\}$.

\subsubsection{More coordinates: $r^\star$, $t$, $\tilde{\phi}$}
Let us define $r^\star$ by 

$$\frac{dr^\star}{dr}\left(r\right)=\frac{r^2+a^2}{\Delta_{-}(r)}, \quad r^\star(r=+\infty)=\pi/2.$$
where $\Delta_{-}(r)$ is given by \eqref{def:deltam}. Note that $r^\star(r_+)=-\infty$. 

By a small abuse of notation, we shall often write for functions $f$ and $g$, $f(r^\star)=g(r)$, instead of $f(r^\star)=g(r(r^\star))$ or $f(r^\star(r))=g(r)$.

Finally, let $r_{cut} = r_+ + \frac{r_{max}-r_+}{2}$, with $r_{max}>r_+$ defined in Lemma \ref{lem:popro} depending only on the parameters $M$, $a$ and $l$, and 
$\chi\left(r\right)$ be a smooth cut-off function with the following property
\begin{equation} \label{cofdef}
\chi\left(r\right) = \left\{ 
\begin{array}{rl} 
1 & \text{if } r \in \left[r_+, \frac{r_{cut}-r_+}{2}\right],\\ 
0 & \text{if } r \geq r_{cut}. \\ 
\end{array} \right.
\end{equation}
We introduce the time coordinate $t$ and another angular coordinate $\tilde{\phi}$ as 
\begin{equation} \label{tpdef}
t = t^\star - A\left(r\right) \chi\left(r\right) \textrm { \ \ \ and \ \ \ } \tilde{\phi} = \phi - B\left(r\right) \chi\left(r\right)\, , \nonumber
\end{equation}
where 
\begin{equation}
\frac{dA}{dr} = \frac{2Mr}{\Delta_- \left(1+\frac{r^2}{l^2}\right)}\, , \textrm{ \ \ \ \ \ \ \ \ } \frac{dB}{dr} = \frac{a \left(1-\frac{a^2}{l^2}\right)}{\Delta_-} \nonumber
\end{equation}
and $A$ and $B$ vanish at infinity.

Note that $t$, $\tilde{\phi}$ and $r^\star$ are not well behaved functions at the horizon $r=r_+$. As a consequence, the coordinate systems $(t,r,\theta,\tilde{\phi })$ and $(t,r^\star,\theta,\phi)$ only cover $\mathrm{int}\left(\mathcal{R}\right)$. Observe also that the two coordinate systems $\left(t,r,\theta,\tilde{\phi}\right)$ and $\left(t^\star,r,\theta,\phi\right)$ are identical for $r \geq r_{cut}$.



\subsubsection{The Kerr-AdS metric for fixed $(a,M,l)$}
We may now introduce the Kerr-AdS metric as the unique smooth extension to $\mathcal{R}$ of the tensor given in the Boyer-Lindquist chart by:
\begin{align}\label{eq:metricbl}
g_{KAdS} = \frac{\Sigma}{\Delta_-} dr^2 + \frac{\Sigma}{\Delta_\theta} d\theta^2 + \frac{\Delta_\theta \left(r^2+a^2\right)^2 - \Delta_- a^2 \sin^2 \theta}{\Xi^2 \Sigma} \sin^2 \theta d\tilde{\phi}^2 \nonumber \\
- 2 \frac{\Delta_\theta \left(r^2+a^2\right)-\Delta_-}{\Xi \Sigma}a \sin^2 \theta \ d\tilde{\phi} dt - \frac{\Delta_- - \Delta_\theta a^2 \sin^2 \theta}{\Sigma} dt^2
\end{align}
where $\Delta_-$ is defined by \eqref{def:deltam} and 
\begin{align}
\Sigma = r^2 + a^2 \cos^2 \theta, 
\ \ \ \ \ 
\Delta_\theta = 1 - \frac{a^2}{l^2} \cos^2 \theta, \label{def:deltatheta}
 \ \ \ \ \ 
\Xi = 1 - \frac{a^2}{l^2} \, .
\end{align}
See \cite{gs:dpkads} for explicit expressions for the inverse of (\ref{eq:metricbl}).
%
%
%
That the tensor \eqref{eq:metricbl} indeed extends to a smooth metric on $\mathcal{R}$ is clear from expressing the metric in $\left(t^\star,r,\theta,\phi\right)$ coordinates, which is carried out explicitly in \cite{gs:dpkads}. Note that for $a=0$, the metric (\ref{eq:metricbl}) reduces to the well-known Schwarzschild-AdS spacetime
$$
g=-\left(1-\frac{2M}{r}+\frac{r^2}{l^2}\right)dt^2+\left(1-\frac{2M}{r}+\frac{r^2}{l^2}\right)^{-1}dr^2+r^2d\sigma_{S^2} \, ,
$$
where $d\sigma_{S^2}$ is the standard metric on the unit sphere.


\subsection{The norms} \label{se:wpbs}  \label{sec:norms}
Let $\slashed{g}$ and $\slashed{\nabla}$ denote the induced metric and the covariant derivative on the spheres $S^2_{t^\star, r}$ of constant $t^\star$ and $r$ in $\mathcal{R}$.

 We write $|\slashed{\nabla} ... \slashed{\nabla} \psi|^2=\slashed{g}^{AA^\prime} ... \slashed{g}^{BB^\prime} \slashed{\nabla}_A ... \slashed{\nabla}_{B} \psi  \slashed{\nabla}_{A^\prime} ... \slashed{\nabla}_{B^\prime} \psi$ to denote the induced norms on these spheres.
 
 We denote by $\Omega_i$ with $i=1,2,3$ the standard basis of angular momentum operators on the unit sphere in $\theta$, $\phi$ coordinates. 
 
 With these conventions, we define the energy densities
 \begin{align} \label{densities}
e_1 \left[\psi\right]  &= \frac{1}{r^2} \left(\partial_{t^\star} \psi\right)^2 + r^2 \left(\partial_r \psi\right)^2 + |\slashed{\nabla} \psi|^2 + \psi^2 \, , \nonumber \\
\tilde{e}_2 \left[\psi\right] &= e_1\left[\psi\right] + e_1\left[\partial_{t^\star} \psi\right] + r^4 \left(\partial_{r} \partial_{r} \psi\right)^2 + r^2 | \partial_r \slashed{\nabla} \psi|^2 + |\slashed{\nabla} \slashed{\nabla} \psi|^2 \, , \nonumber \\
{e}_2 \left[\psi\right]  &=\tilde{e}_2 \left[\psi\right] + \sum_i e_1\left[\Omega_i \psi\right]  \, .
\end{align}

Similarly, we define the following energy norms for the scalar field $\psi$, cf.~\cite{Holzegelwp}:
\begin{eqnarray}
\| \psi \|^2_{H^{0,s}_{AdS}\left(\Sigma_{t^\star}\right)} &=&  \int_{\Sigma_{t^\star}} r^s  \psi^2 r^2 dr \sin \theta d\theta d{\phi}\nonumber \\
\| \psi \|^2_{H^{1,s}_{AdS}\left(\Sigma_{t^\star}\right)} &=& 
\int_{\Sigma_{t^\star}} r^s \left[ r^2 \left(\partial_r \psi \right)^2 + | \nabb \psi |^2 + \psi^2 \right] r^2 dr \sin \theta d\theta d{\phi}\nonumber \\
\| \psi \|^2_{H^{2,s}_{AdS}\left(\Sigma_{t^\star}\right)} &=& \| \psi \|^2_{H^{1,s}_{AdS}\left(\Sigma_{t^\star}\right)} \nonumber \\
&+& \int_{\Sigma_{t^\star}} r^s \bigg[  r^4 \left(\partial_r \partial_r \psi \right)^2 
+r^2 | \nabb \partial_r \psi|^2 +| \nabb \nabb \psi |^2 \bigg] r^2 dr \sin \theta d\theta d{\phi}\, . \nonumber
\end{eqnarray}
Note in particular the relations
\begin{align}
\int_{\Sigma_{t^\star}} e_1 \left[\psi\right] \ r^2 dr \sin \theta d\theta d\phi = \| \psi \|^2_{H^{1,0}_{AdS}\left(\Sigma_{t^\star}\right)} + \| \partial_{t^\star} \psi \|^2_{H^{0,-2}_{AdS}\left(\Sigma_{t^\star}\right)} \, ,
\end{align}
\begin{align}
\int_{\Sigma_{t^\star}} e_2 \left[\psi\right]  \ r^2 dr \sin \theta d\theta d\phi = \| \Omega_i \psi \|^2_{H^{1,0}_{AdS}\left(\Sigma_{t^\star}\right)} + \| \partial_{t^\star} \psi \|^2_{H^{1,0}_{AdS}\left(\Sigma_{t^\star}\right)}  \nonumber \\ 
+  \| \psi \|^2_{H^{2,0}_{AdS}\left(\Sigma_{t^\star}\right)} + \| \partial_{t^\star} \partial_{t^\star} \psi \|^2_{H^{0,-2}_{AdS}\left(\Sigma_{t^\star}\right)} \, .
\end{align}

Higher order norms may be defined similarly. We denote by $H^{k,s}_{AdS}(\Sigma_{t^\star})$, the space of functions $\psi$ such that $ \nabla^i \psi \in L^2_{loc}(\Sigma_{t^\star})$ for $i=0,..., k$ and $\| \psi \|^2_{H^{k,s}_{AdS}\left(\Sigma_{t^\star}\right)} < \infty$.
We denote by $CH_{AdS}^{k,s}$ the set of functions $\psi$ defined on $\mathcal{R}$, such that 
\begin{align}
\psi \in &\bigcap_{q=0,..,k} C^q\left(\mathbb{R}_{t^\star};H^{k-q,s_q}_{AdS}(\Sigma_{t^\star})\right) \nonumber \\
&\textrm{where $s_k = -2$, $s_{k-1} = 0$ and $s_j = s$ for $j=0, ... , k-2$. }\nonumber
\end{align}
When $s=0$, we will feel free to drop the $s$ in the notation, i.e. $H^{k,0}_{AdS}:=H^{k}_{AdS}$ and $CH_{AdS}^{k}:=CH_{AdS}^{k,0}$. 

\subsection{A Final Remark}
In \cite{gs:dpkads}, the coordinates $t$ and $\tilde{\phi}$ in (\ref{tpdef}) are defined with the $\chi\left(r\right)$ of (\ref{cofdef}) being \emph{globally} equal to $1$. Here, for convenience in the subsequent analysis (which happens mostly away from the horizon, in $r\geq r_{max}$), we have altered these coordinates \emph{away from the horizon} to agree with the Boyer-Lindquist coordinates. Note that these two coordinate systems are equivalent in the sense that the statement ``$\| \psi \|^2_{H^{k,s}_{AdS}\left(\Sigma_{t^\star}\right)}$ decays logarithmically in $t^\star$" is independent of whether the coordinates (and $\Sigma_{t^\star}$-slices) of \cite{gs:dpkads} or the cut-off coordinates (\ref{tpdef}) are used.

\section{Separation of variables and reduced equations} \label{se:sere}
\subsection{The (modified) oblate spheroidal harmonics} \label{se:sphhar}
For each $\xi \in \mathbb{R}$, define the unbounded $L^2\left(\sin \theta d\theta d\tilde{\phi}\right)$-self-adjoint operator $P_{\mathbb{S}^2}\left(\xi\right)$ with domain being $H^2\left(\mathbb{S}^2\right)$-complex valued functions (see Section 7 of \cite{DafRodsmalla} for a more detailed discussion) as
\begin{eqnarray}\label{def:pop}
-P_{\mathbb{S}^2}\left(\xi\right) f &=& \frac{1}{\sin \theta} \partial_\theta \left(\Delta_\theta \sin \theta \partial_\theta f \right) + \frac{\Xi^2}{\Delta_\theta} \frac{1}{\sin^2 \theta} \partial_{\tilde{\phi}}^2 f \nonumber \\ 
&&\hbox{}+ \Xi \frac{\xi^2}{\Delta_\theta} \cos^2 \theta f - 2 i  \xi \frac{\Xi}{\Delta_\theta} \frac{a^2}{l^2} \cos^2 \theta \  \partial_{\tilde{\phi}}f \, \,.
\end{eqnarray}
We also define the operator $P_{\alpha}$, which is equal to
\begin{equation} \label{def:palpha}
P_{\mathbb{S}^2,\alpha} \left(\xi\right) = \left\{
\begin{array}{rl} 
P \left(\xi\right) + \frac{\alpha}{l^2}a^2 \sin^2 \theta & \text{if } \alpha > 0 \, , \\
P\left(\xi\right) + \frac{|\alpha|}{l^2}a^2 \cos^2 \theta & \text{if } \alpha \leq 0 \, .
\end{array} \right. 
\end{equation}
For $l \rightarrow \infty$ the operator $P_{\mathbb{S}^2}(\xi)$ reduces to an oblate spheroidal operator
on $\mathbb{S}^2$ as considered for instance in \cite{DafRodsmalla}. If also $\xi=a=0$ we retrieve the Laplacian on the round sphere. Both operators $P_{\mathbb{S}^2}(\xi)$ and $P_{\mathbb{S}^2,\alpha}(\xi)$ have discrete spectrum. The eigenvalues and eigenvector of $P_{\mathbb{S}^2,\alpha}(\xi)$ will be denoted as follows:
\begin{align} 
P_{\mathbb{S}^2,\alpha} \left(\xi\right) \textrm{ \ has eigenvalues }\lambda_{m \ell}\left(\xi\right) \textrm{ with eigenfunctions $S_{m\ell} \left(\xi, \cos \theta \right) e^{im \tilde{\phi}}$}. \nonumber
\end{align}
Later we will restrict attention to axisymmetric solutions and hence to the eigenvalues $\lambda_{0 \ell}$ of the operators
\begin{align} \label{angprob}
-P_{\theta,\alpha}\left(\xi^2\right) := -P_{\mathbb{S}^2,\alpha}\left(\xi\right) \Big|_{m=0} \, .
\end{align}
In fact, it will be convenient (in view of their manifest positivity) to work with the eigenvalues of the shifted operators (for $\alpha<0$)
\begin{align} \label{auxfo}
-P_{\theta,\alpha} \left(\xi^2\right) -\xi^2 = &\frac{1}{\sin \theta} \partial_\theta \left(\Delta_\theta \sin \theta \partial_\theta \left( \, \cdot \, \right) \right) -\frac{|\alpha|}{l^2} a^2 \cos^2 \theta - \frac{\sin^2 \theta}{\Delta_\theta} a^2  \xi^2 \, 
\end{align}
and with the $\cos^2$ replaced by a $\sin^2$ in the second term in case that $\alpha \leq 0$.
In view of these considerations, we shall denote the eigenvalues of the operator $P_{\theta,\alpha} \left(\xi^2\right) + \xi^2$ by ${\mu}_{\ell} \left(\xi^2\right)$. Note that by min-max and comparison with the spherical Laplacian (see Lemma 5.1 of \cite{gs:dpkads}) in the axisymmetric case we have
\begin{align} \label{minmaxangular}
\mu_{\ell} \left(\xi^2\right) \geq \mu_{\ell} \left(0\right) \geq \Xi \ell \left(\ell+1\right) > c_{a,l} \ell\left(\ell+1\right) \, .
\end{align}

\subsection{The separation of variables} \label{se:sepv}
In this section, we present the reduced equations obtained after separation of variables. For the \emph{construction} of quasimodes, it would be sufficient to start directly from the reduced equations. However, to explain their relation with (\ref{mwe}), we will instead derive them from the Klein-Gordon equation \eqref{mwe}. Thus, in this section $\psi$ will denote any regular solution to \eqref{mwe}. 

Let us introduce the time-Fourier transform\footnote{Note that in general $\psi$ is not an $L^2$ function in time and therefore, $\widehat{\psi}$ is defined only as a tempered distribution. Since here we are merely trying to justify the origin of the reduced equations, it will be sufficient to understand all computations formally.}
\begin{align} \label{foutime}
\psi \left(t,r,\theta,\tilde{\phi}\right) = \frac{1}{\sqrt{2\pi}} \int_{-\infty}^\infty e^{-i\omega t} \widehat{\psi} \left(\omega,r,\theta,\tilde{\phi}\right) d\omega \, .
\end{align}
 We decompose the $\widehat{\psi}$ of equation (\ref{foutime}) as
\begin{align}
\widehat{\psi}\left(\omega, r , \theta, \tilde{\phi}\right)  &= \sum_{m\ell} \left(\widehat{\psi}\right)^{(a\omega)}_{m\ell} \left(r\right) S_{m\ell} \left(a \omega, \cos \theta\right) e^{im\tilde{\phi}} \, ,
\end{align}
where $S_{m\ell}$ are the modified spheroidal harmonics introduced precedently and 
\begin{align}
 \left(\widehat{\psi}\right)^{(a\omega)}_{m\ell} \left(r\right) = \frac{1}{\sqrt{2\pi}} \int_{-\infty}^\infty dt \int_{\mathbb{S}^2\left(t,r\right)} d\tilde{\sigma} e^{i\omega t} \ S_{m\ell} \left(a \omega, \cos \theta\right) e^{-im\tilde{\phi}} \  \psi\left(t,r,\theta,\tilde{\phi}\right), \nonumber
\end{align} 
with $d\tilde{\sigma} = \sin \theta d\theta d\tilde{\phi}$.

After the renormalization
\begin{align}
u^{(a\omega)}_{m\ell} \left(r\right)  = \left(r^2+a^2\right)^\frac{1}{2} \left(\widehat{\psi}\right)^{(a\omega)}_{m\ell} \left(r\right) \,, 
\end{align}
we finally obtain from \eqref{mwe} the equation
\begin{align}
\left[ u^{(a\omega)}_{m\ell} \left(r\right)\right]^{\prime \prime} + \left(\omega^2 - V^{(a\omega)}_{m\ell} \left(r\right) \right) u^{(a\omega)}_{m\ell} =0\label{eq:uode} \, ,  
\end{align}
where the potential $V^{(a\omega)}_{m\ell} \left(r\right)$ is defined as
\begin{align}
V^{(a\omega)}_{m\ell} \left(r\right) = V^{(a\omega)}_{+, m \ell}\left(r\right) + V^{(a\omega)}_{0, m \ell} \left(r\right) +V_\alpha \left(r\right) \, ,
\end{align}
where
\begin{eqnarray} \label{eq:v+}
V^{(a\omega)}_{+, m \ell}\left(r\right) &=& -\Delta_-^2 \frac{3 r^2 }{\left(r^2 + a^2\right)^4} + \Delta_-\frac{5\frac{r^4}{l^2} + 3r^2 \left(1+\frac{a^2}{l^2}\right) - 4Mr + a^2}{\left(r^2 + a^2\right)^3} \nonumber \\
&=& \left(r^2+a^2\right)^{-\frac{1}{2}} \left(\sqrt{r^2+a^2}\right)^{\prime \prime},\\
V^{(a\omega)}_{0, m \ell} \left(r\right)&=& \frac{\Delta_- \left(\lambda_{m\ell} + \omega^2 a^2\right)-\Xi^2 a^2 m^2 - 2m\omega a \Xi \left(\Delta_- - \left(r^2+a^2\right)\right)}{\left(r^2 + a^2\right)^2}\nonumber \\
V_\alpha \left(r\right) &=&-\frac{\alpha}{l^2} \frac{\Delta_-}{\left(r^2+a^2\right)^2} \left(r^2 + \Theta\left(\alpha\right) a^2 \right) \label{eq:va} \, ,
\end{eqnarray}
with $\Theta \left(\alpha\right)=1$ for $\alpha > 0$ and $\Theta \left(\alpha\right)=0$ if $\alpha\leq 0$ (recall that the $\lambda_{m\ell}$ also depend on $\alpha$ through (\ref{def:palpha})). Note that $V_+$ grows like $\frac{2r^2}{l^4}$ near infinity while the $V_0$-part remains bounded.
\subsection{The axisymmetric reduced equations} \label{sec:potential}
We now look at the axisymmetric case, i.e.~we consider the above equations under the assumption that $\psi$ is independent of the azimuthal variable $\tilde{\phi}$. The reduced equations are then obtained by setting $m=0$ in the decomposed equations. Hence, we will consider the following set of equations:
\begin{align}
\left[ u^{(a\omega)}_{0\ell} \left(r\right)\right]^{\prime \prime} + \left(\omega^2 - V^{(a\omega)}_{0\ell} \left(r\right) \right) u^{(a\omega)}_{0\ell} \left(r\right)  =0\label{eq:red} \, ,  
\end{align}
where the potential $V^{(a\omega)}_{0\ell} \left(r\right)$ is defined as
\begin{align}
V^{(a\omega)}_{0\ell} \left(r\right) = V_{junk}\left(r\right) +V_{mass} \left(r\right) + V_{\sigma} \left(r\right) \cdot \mu_\ell\left(\omega^2 a^2\right) \, ,
\end{align}
where
\begin{eqnarray} \label{eq:av+}
V_{junk}\left(r\right) &=& -\Delta_-^2 \frac{3 r^2 }{\left(r^2 + a^2\right)^4} + \Delta_-\frac{5\frac{r^4}{l^2} + 3r^2 \left(1+\frac{a^2}{l^2}\right) - 4Mr + a^2}{\left(r^2 + a^2\right)^3}   \nonumber \\
&&\hbox{}-\frac{2}{l^2} \frac{\Delta_- r^2 }{\left(r^2+a^2\right)^2}  
-\frac{\alpha}{l^2} \frac{\Delta_- }{\left(r^2+a^2\right)^2}\Theta\left(\alpha\right) a^2 \, ,
\nonumber \\
V_{mass} \left(r\right) &=&\frac{2-\alpha}{l^2} \frac{\Delta_- r^2}{\left(r^2+a^2\right)^2} \, , \nonumber \\
V_{\sigma}\left(r\right) &=& \frac{\Delta_-}{\left(r^2+a^2\right)^2} .\label{eq:ava} \,
\end{eqnarray}
Here we rearranged the terms in the different potentials so that $V_{mass}=0$ corresponds to the conformal case $\alpha=2$. In particular, we have $| V_{junk}\left(r\right)| \leq C_{M,l,a} \frac{\Delta_-}{\left(r^2+a^2\right)^2}\lesssim V_\sigma$ and hence that $V_{junk}$ is uniformly bounded.

The properties of $V_{\sigma}$ are summarized in the following lemma.

\begin{lemma}[Properties of $V_\sigma$]\label{lem:popro}
For all $|a| < l$, the potential $V_\sigma$ enjoys the following properties:
\begin{itemize}
\item $V_\sigma(r^\star) \rightarrow \frac{1}{l^2}$ as $r^\star \rightarrow \pi/2$.
\item $V_\sigma$ has a unique local and global maximum $V_{max}$ at $r^\star_{\max}$ in $[r^\star_+,\pi/2)$. Also, $V_{\sigma}$ is monotonically decreasing in $[r^\star_{max},\pi/2)$. 
\item $V_{\max}=V_\sigma(r^\star_{\max}) \ge \frac{1}{l^2} + \frac{\frac{3M}{\Xi} + \Xi a^2}{\left(\left(\frac{3M}{\Xi}\right)^2+a^2\right)^2} $.
\end{itemize}
\end{lemma}
\begin{remark}In particular, for any $0< a_{0} < l$, and for all $|a| \le a_{0}$, the size of the interval $[1/l^2, V_\sigma(r^\star_{\max})]$ is uniformly (in $a$) bounded from below by a strictly positive uniform constant.
\end{remark}
\begin{proof}The first claim can be trivially checked. 
For the second and third claims, let us write
\begin{eqnarray}
V_\sigma\left(r\right) = \frac{\Delta_-}{\left(r^2+a^2\right)^2} &=&  \frac{\left(r^2+a^2\right)\left(\frac{r^2}{l^2}+\frac{a^2}{l^2}+ \Xi\right) - 2Mr}{\left(r^2+a^2\right)^2}\nonumber \\
 &=& \frac{1}{l^2} + \frac{\Xi\left(r^2+a^2\right) - 2Mr}{\left(r^2+a^2\right)^2}. \label{eq:fs}
\end{eqnarray}

The $r$ derivative of $V_\sigma$ is
\begin{align}
\partial_r V_\sigma \left(r\right) &= \frac{\left(r^2+a^2\right)^2 \left(2\Xi r - 2M\right) - 4r \left(r^2+a^2\right) \left(\Xi\left(r^2+a^2\right) - 2Mr\right)}{\left(r^2+a^2\right)^4} \nonumber \\
&=  \frac{\left(r^2+a^2\right) \left(2\Xi r - 2M\right) - 4r  \left(\Xi\left(r^2+a^2\right) - 2Mr\right)}{\left(r^2+a^2\right)^3} \nonumber \\
&=\frac{-2\Xi r^3 + 6M r^2 - 2\Xi a^2 r - 2Ma^2}{\left(r^2+a^2\right)^3} 
\end{align}
If $a=0$, then one has as usual that the only zero in $[r_+,\infty)$ is at $r=3M$. Assume $a\neq 0$. Observe that the derivative is positive near negative infinity, negative at $r=0$, positive at $r=r_+$ and negative near infinity. This tells us that there are three real roots for $\partial_r V_\sigma$. The one of interest to us is the one corresponding to the (unique) maximum of $V_\sigma$ in the interval $\left[r_+,\infty\right)$.  Define $r_{guess} = \frac{3M}{\Xi}$. We claim $r_+ < r_{guess} < \infty$. This follows from $r_+ < 2M$, which is in turn a consequence of the fact that $\Delta_- = r \left(r-2M\right) + \frac{r^4}{l^2} + a^2 \left(1+\frac{r^2}{l^2}\right)$ is positive for $r\geq 2M$. Finally, at $r_{guess} = \frac{3M}{\Xi}$ we have from (\ref{eq:fs})
\[
V_\sigma\left(r_{guess}\right)=\frac{1}{l^2} + \frac{\frac{3M}{\Xi} + \Xi a^2}{\left(r_{guess}^2+a^2\right)^2} > \frac{1}{l^2} \, .
\]
\end{proof}

\section{Bound states}\label{se:te}
As proven in our  \cite{gs:dpkads}, there exist no periodic solutions of the massive wave equation on Kerr-AdS. In this section, we will introduce an additional boundary, located at $r=r_{max}$ (the location of the top of the potential $V_\sigma$, as defined in Lemma \ref{lem:popro}), enabling us to construct periodic solutions whose associated energies lie below the top of $V_\sigma$.

In order to avoid confusion between the mode number $\ell$ and the real number $l$ determining the cosmological constant, we introduce the semi-classical parameter $h>0$ by defining
\begin{align}
h^{-2}(\ell,\omega):=\lambda_{0\ell}(a^2 \omega^2) + a^2 \omega^2=\mu_\ell\left(a^2\omega^2\right) \, ,
\end{align}
as well as the shorthand
\begin{align}
h_0^{-2} = h^{-2}\left(\ell,0\right) = \mu_{\ell}\left(0\right) \, .
\end{align}
In the rest of this paper, both notations, $u_\ell$ or $u_h$, will be used when we want to make explicit that a solution $u$ depends on $h$ or $\ell$.

Having in mind a semi-classical type analysis with semi-classical parameter $h$, we then rewrite $\eqref{eq:red}$ as \\ \\
{\bf Non-linear Eigenvalue Problem:}
\begin{align}\label{eq:semred}
\begin{split}
h^2 \omega^2 u = P\left(h\right)u:= -h^2 u''+ V_\sigma u + h^2(V_{junk}+V_{mass})u \, , \\
\textrm{with boundary cond.} \ \ u\left(r^\star_{max}\right)=0 \ \ \text{and} \ \ \int_{r^\star_{max}}^{\pi/2} \left[|u^\prime|^2 + r^2 u^2 \right] dr^\star < \infty \, .
\end{split}
\end{align}
%
\begin{remark}
We will refer to the above boundary conditions as Dirichlet conditions. Note that they imply that $u \left(\pi/2\right)=0$. See Remark \ref{rem:bvp} below.
\end{remark}
Unless $a=0$, (\ref{eq:semred}) is a non-linear eigenvalue problem. Indeed, a solution to the eigenvalue problem 
\begin{equation}\label{eq:eig}
\kappa u = P\left(h\right)u
\end{equation} with Dirichlet boundary conditions is a solution of (\ref{eq:semred}) if and only if $\kappa=h^2 \omega^2$. In case $a=0$, (\ref{eq:semred}) reduces to the linear problem (\ref{eq:eig}) since the $P\left(h\right)$ operator becomes independent of $\omega$, and, thefore, given any solution to (\ref{eq:eig}), one simply obtains a solution to \eqref{eq:semred} by defining $\omega^2= h^{-2} \kappa$. 

What we would like to prove is that -- given fixed parameters $M$, $a$, $l$, $\alpha$ --  we can find, for any sufficiently small $h$ (or equivalently sufficiently large $\ell$), an $\omega_\ell^2$ such that (\ref{eq:semred}) is solved for some $u_\ell$, and to control the size of this $\omega_\ell^2$ (cf.~Proposition \ref{prop:eig2} below).

In order to understand (\ref{eq:semred}),
it will be useful to first look at the following \\ \\
{\bf Linear Eigenvalue Problem:}
\begin{align}\label{eq:semred2}
\begin{split}
h_0^2 \omega^2 \ u= P_{base}\left(h_0\right)u := -h_0^2  u''+ V_\sigma u + h_0^2(V_{junk}+V_{mass})u \, , \\
\textrm{with boundary cond.} \ \ u\left(r^\star_{max}\right)=0 \ \ \text{and} \ \ \int_{r^\star_{max}}^{\pi/2} \left[|u^\prime|^2 + r^2 u^2 \right] dr^\star < \infty \, .
\end{split}
\end{align}
\newline
As explained above, (\ref{eq:semred2}) can be seen as a linear eigenvalue problem because $h_0$ and  therefore $P_{base}\left(h_0\right)$ depends only on $\ell$ (but not on $\omega$). 

\begin{remark} \label{rem:bvp}
By Proposition 4 of \cite{CWgh}, (\ref{eq:semred2}) is a well-posed eigenvalue problem. This is a non-trivial statement because the potential $V_{mass}$ is unbounded on the domain $\left(r^\star_{max},\pi/2\right)$ unless $\alpha=2$. The condition $\int_{r^\star_{max}}^{\pi/2} \left[|u^\prime|^2 + r^2 u^2 \right] dr^\star < \infty$ implies that $\psi = \sqrt{r^2+a^2}u \cdot S_{0\ell}\left(\theta\right) \in H^1_{AdS}$ (in particular $r^{1/2-\epsilon} u\left(\pi/2\right)=0$ for any $\epsilon>0$) and ensures the existence of a positive discrete spectrum with eigenfunctions in the energy space.\footnote{Using the twisted derivatives of \cite{CWgh}, i.e.~writing $u^{\prime \prime} + V_{mass} u = r^{n} \left( r^{-2n} \left(r^n u\right)^\prime\right)^\prime + V_{twist} \cdot u$ with $n=1/2-1/2\sqrt{9-4\alpha}$ and hence $V_{twist}$ uniformly bounded, one could generalize the construction of the paper to other boundary conditions.}
\end{remark}

To prove that eigenvalues $h_0^2 \omega^2$ for (\ref{eq:semred2}) exist in a suitable range, we will perform a semi-classical type analysis for a semi-classical operator whose principal part should be  $-h_0^2 u''+ V_\sigma u$. Since $h_0^2 V_{junk}$ is controlled by $h_0^2V_\sigma$, this term will be lower order and hence negligible. On the other hand, unless one considers the conformal case $\alpha=2$, the a-priori lower order (in powers of $h_0$) potential term  $h_0^2 \cdot V_{mass}$ is unbounded near $r^\star=\pi/2$, so that some care (a Hardy inequality) is required.

Observe finally that if we set $a=0$ in $V_{\sigma}$, $V_{junk}$ and $V_{mass}$ in (\ref{eq:semred2}), then
\begin{align} \label{ssp}
P_{base}^{a=0} \left(h_0\right) u = \kappa \cdot u
\end{align}
would be precisely the eigenvalue problem one needs to study for Schwarzschild-AdS. In any case, in the next section we will establish the existence of eigenvalues $\kappa = h^2_0 \omega^2$ of the more general eigenvalue problem  (\ref{eq:semred2}) as the latter is easier to connect to the full problem  (\ref{eq:semred}).

\subsection{Weyl's law for the linear eigenvalue problem  (\ref{eq:semred2})} \label{sec:Weyllin}
For the purpose of this section, it will be convenient to introduce the following notation. For all $c < d$, we define $P^{base}_{DD}(c,d)$ as the following eigenvalue problem
\begin{eqnarray}
&& P_{base}(h_0) \ u = \kappa \ u \, , \nonumber 
\end{eqnarray}
with Dirichlet boundary conditions $u(c)=u(d)=0$ and -- in case that $d=\pi/2$ -- the condition $\int_c^d dr^\star \left[|u^\prime|^2 + r^2\cdot u^2 \right] < \infty$ (cf.~Remark \ref{rem:bvp}), where $P_{base}(h_0)$ is the operator defined by (\ref{eq:semred2}). Similarly, we will write $P^{base}_{NN}(c,d)$ for the Neumann problem (which will never be considered with $d=\pi/2$). Finally, we write $P^{base}_{ND}(c,d)$ for Neumann boundary at $c$ and Dirichlet at $d$. In the latter case we again impose $\int_c^d dr^\star \left[|u^\prime|^2 + r^2\cdot u^2 \right] < \infty$ in case that $d=\pi/2$ to ensure that all eigenfunctions live in the energy space, cf.~Remark \ref{rem:bvp}. Note that $P^{base}_{DD}(r^\star_{\max},\pi/2)$ is precisely the linear eigenvalue problem (\ref{eq:semred2}).

The aim of this section is to establish the following proposition. 

\begin{proposition}\label{prop:eig} 
Let $\alpha<\frac{9}{4}$, $M>0$ and $|a|<l$ be fixed, and $E \in \left(\frac{1}{l^2}, V_{\max}\right)$ be given.  Then, for any $\delta > 0$ such that $\left[E-\delta,E+\delta\right] \subset \left(\frac{1}{l^2}, V_{\max}\right)$, there exists an $H_{0} > 0$, such that the following statement holds. For any $0 < h_0 \le H_0$, there exists a smooth solution $u_{h_0}$ of the eigenvalue problem $P^{base}_{DD}(r^\star_{\max},\pi/2)$,  with corresponding eigenvalue $\kappa$ lying in $\left[E-\delta,E+\delta\right]$. In particular, there exists a sequence $\left((h_{0})_n, u_{(h_{0})_n}\right)$ such that the associated eigenvalues $\kappa \left((h_{0})_n \right) \rightarrow E$ as $(h_{0})_n \rightarrow 0$.
\end{proposition}
In the rest of this section, $(a,M,l,\alpha)$ are fixed parameters satisfying the assumptions of the proposition. 

We shall in fact prove in this section a stronger result than Proposition \ref{prop:eig}, namely a version of Weyl's law adapted to our problem. This is the statement of Lemma \ref{lem:Weylslaw}, from which Proposition \ref{prop:eig} immediately follows. The proof of Lemma \ref{lem:Weylslaw} in turn requires the following auxiliary lemma, which ensures non-existence of eigenvalues below a certain threshold. 

\begin{lemma} \label{lem:absence}
Let $E> 0$ be given. Then there exists an $H_0> 0$ so that for all $0< h_0 \le H_0$, there exists a $r_K^\star(E,h_0,\alpha)$ such that the problems $P^{base}_{DD}(r_K^\star,\pi/2)$ and $P^{base}_{ND}(r_K^\star,\pi/2)$ have no solutions with $\kappa \le E$. Moreover, 
\begin{align} \label{rkest}
r_{\max} < r_K \le \frac{C}{h_0}\cdot E ,
\end{align}
where $C$ depends only $M,a,l$ and $\alpha$.
\end{lemma}
\begin{proof}
Assume there was a solution $u$ of $P^{base}_{DD}(r_K^\star,\pi/2)$ or $P^{base}_{ND}(r_K^\star,\pi/2)$ with $\kappa\leq E$. Then we would have
\begin{align} \label{tui}
\int_{r_K^\star}^{\pi/2}  dr^\star \left[h_0^2 |u^\prime|^2 + \left(V_{\sigma}+h_0^2V_{junk}+h_0^2V_{mass} - E \right)|u|^2 \right] \leq 0
\end{align}
for this $u$. On the other hand, the Hardy inequality
\begin{align} \label{Hardy}
\int_{r_K^\star}^{\pi/2} dr^\star \frac{\Delta_-}{r^2+a^2} |u|^2 
\leq 4l^2 \int_{r_K^\star}^{\pi/2} dr^\star  |u^\prime|^2
\end{align}
proven in our  \cite{gs:dpkads} holds for $u$. This implies that
\[
\int_{r_K^\star}^{\pi/2}  dr^\star \left[ \left( h_0^2\frac{1}{4 l^2} \frac{\Delta_-}{r^2+a^2} + V_{\sigma}+h_0^2V_{junk}+h_0^2V_{mass} - E \right)|u|^2 \right] \leq 0 \, .
\]
The dominant term in the integrand near infinity is $h_0^2 \left(\frac{9}{4}-\alpha\right) \frac{r^2}{l^2}$ which is positive, while all other terms remain bounded. Hence by choosing $r^\star_K$ sufficiently large ($r_K \geq \frac{C}{h_0} \cdot E$ for some constant $C$) we obtain a contradiction as the round bracket in the integrand eventually becomes positive.
\end{proof}

Consider now the eigenvalue problem $P^{base}_{DD}\left(r^\star_{max},\pi/2\right)$ and fix an energy level $\mathcal{E} \in \left(\frac{1}{l^2},V_{max}\right)$. Lemma \ref{lem:absence} produces an $r^\star_K \left(\mathcal{E},h_0, \alpha\right)$, to which we associate the phase-space volume
\begin{align} \label{Pedh}
\mathcal{P}_{\mathcal{E},h_0,\alpha} &= Vol\{ \ \left(\xi, r^\star\right) \in \mathbb{R} \times \left[r^\star_{max}, r^\star_K\right] \ | \ \  \xi^2 + V_\sigma + h_0^2 V_{mass} + h_0^2 V_{junk} \leq \mathcal{E} \} \nonumber \\
&=
2 \int_{r^\star_{max}}^{r^\star_K} dr^\star \sqrt{\mathcal{E} - V \left(r^\star\right)} \cdot  \chi_{\{V\left(r^\star\right) \leq \mathcal{E}\}} \, .
\end{align}
Note that for fixed $\mathcal{E}$, this expression converges uniformly in $h_0$ as $h_0 \rightarrow 0$. This is already immediate for $\alpha \leq 2$:  $V\left(r^\star\right)$ is then bounded below and hence the integrand itself is obviously uniformly bounded in $h_0$. For $2< \alpha < 9/4$, the integral (\ref{Pedh}) also converges uniformly in $h_0$ since
\[
 \int_{r^\star_{max}}^{r^\star_K} dr^\star h_0 r \leq C  \int_{r_{max}}^{r_K} dr h_0 \frac{1}{r}  \leq C h_0 \log \left[ \frac{C}{h_0} E \right] 
\]
goes to zero as $h_0 \rightarrow 0$. Here we have used the estimate (\ref{rkest}) on $r_K$. 

Finally, to state and prove Weyl's law, we also introduce an expression for the phase space volume between two energy levels, say $\left[E-\delta, E+\delta\right] \subset \left(\frac{1}{l^2},V_{max}\right)$:
\begin{align}
\mathcal{Q}_{E,\alpha}  &= \lim_{h_0 \rightarrow 0}\mathcal{P}_{E+\delta,h_0,\alpha} -\lim_{h_0 \rightarrow 0}\mathcal{P}_{E-\delta,h_0,\alpha}  \nonumber \\
&= Vol\{ \left(\xi, r^\star\right) \ | \ E-\delta \leq \xi^2 + V_\sigma  \leq E+\delta \} \,. \nonumber
\end{align}
By an elementary computation, we have a lower bound $\mathcal{Q}_{E,\alpha} \geq C_{E,M,l,\alpha}\cdot \delta$ for a constant independent of $h_0$.
\begin{lemma} \label{lem:Weylslaw}
Consider the eigenvalue problem $P^{base}_{DD}\left(r^\star_{max},\pi/2\right)$. Fix an energy level $V_{max}>E>\frac{1}{l^2}$ and prescribe a $\delta>0$ small. Then the number of eigenvalues of $P^{base}_{DD}\left(r^\star_{max},\pi/2\right)$ lying in the interval $\left[E-\delta,E+\delta\right] \subset\left(\frac{1}{l^2},V_{max}\right)$, denoted $N \left[E-\delta, E+\delta\right]$,  satisfies Weyl's law
\begin{align} \label{eqWeyl}
N \left[E-\delta, E+\delta\right] \underset{h_0 \rightarrow 0}{\sim} \frac{1}{2\pi h_0} \mathcal{Q}_{E,\alpha}.
\end{align}
\end{lemma}

\begin{proof}
Choose $r^\star_K\left(E+\delta,h_0\right)$ such that by Lemma \ref{lem:absence} there are no eigenvalues below $ E+\delta$ of $P^{base}_{DD}(r_K^\star,\pi/2)$ and $P^{base}_{ND}(r_K^\star,\pi/2)$. We equipartition the domain $\left[r^\star_{max},r^\star_K\left(E,h_0\right)\right]$ into $k$ intervals of length 
$\beta = \frac{\pi/2-r^\star_K}{k}$. We then consider the following two comparison problems:
\begin{itemize}
\item The Dirichlet problem $P^{base}_{DD}\left(r^\star_K, \pi/2\right)$ in conjunction with $k$ Dirichlet problems  indexed $P^{i+}_D$ ($i=1,...k$) and arising in the following way: They are the problems $P^{base}_{DD}\left(r^\star_{max}, r^\star_{max}+\beta\right), ... ,P^{base}_{DD}\left(r^\star_K-\beta, r^\star_K\right)$ but with the potential replaced by a constant, which equals the \emph{maximum} of the potential on the interval.
\item The mixed problem $P^{base}_{ND}\left(r^\star_K, \pi/2\right)$ in conjunction with $k$ Neumann problems indexed $P^{i-}_N$ ($i=1,...k$) and arising in the following way: They are the problems $P^{base}_{NN}\left(r^\star_{max}, r^\star_{max}+\beta\right), ... ,P^{base}_{NN}\left(r^\star_K-\beta, r^\star_K\right)$ but with the potential replaced by a constant, which equals the \emph{minimum} of the potential on the interval.
\end{itemize}
We can estimate the number of eigenvalues  of $P^{base}_{DD}\left(r^\star_{max},\pi/2\right)$ below a threshold $\mathcal{E}$ by
\begin{align}
&\sum_{i=1}^k N_{\leq\mathcal{E}}\left(P_D^{i+}\right) + N_{\leq \mathcal{E}} \left(P^{base}_{DD}\left(r^\star_K, \pi/2\right)\right) \leq \nonumber \\
&N_{\leq \mathcal{E}} \left(P^{base}_{DD}\left(r^\star_{max},\pi/2\right)\right) \leq \sum_{i=1}^k N_{\leq\mathcal{E}}\left(P_N^{i-}\right) + N_{\leq \mathcal{E}} \left(P^{base}_{ND}\left(r^\star_K, \pi/2\right)\right) \, . \nonumber
\end{align}
By our choice of $r^\star_K$, we have $N_{\leq \mathcal{E}} \left(P^{base}_{DD}\left(r^\star_K, \pi/2\right)\right)=N_{\leq \mathcal{E}} \left(P^{base}_{ND}\left(r^\star_K, \pi/2\right)\right)=0$ for $\mathcal{E}=E+\delta$. On the other hand, for each $P^{i+}_D$ and each $P^{i-}_N$, the number of eigenvalues can be estimated directly (as each problem can be solved explicitly). We have
\begin{eqnarray*}
\sum_{i=1}^k N_{\leq\mathcal{E}}\left(P_D^{i+}\right) &=& \sum_{i=1}^k \Big\lfloor \frac{\beta}{2\pi h_0}\max\left(0,\mathcal{E}-V_+^i\right)\sqrt{\mathcal{E}-V_+^i}\Big\rfloor, \\
&=&\left(\sum_{i=1}^k \frac{\beta}{2\pi h_0}\max\left(0,\mathcal{E}-V_+^i\right)\sqrt{\mathcal{E}-V_+^i}\right)+\mathcal{O}(k). 
\end{eqnarray*}
The estimate for $P_N^{i-}$ is similar with the potential being replaced by $V_-^i$ and the number of eigenvalue in each cell being $\Big\lfloor \frac{\beta}{2\pi h_0}\max\left(0,\mathcal{E}-V_-^i\right)\sqrt{\mathcal{E}-V_-^i}\Big\rfloor+1$. 

To conclude, let us choose the number of cells $k$ such that $k(h_0)$ tends to $\infty$ as $h_0$ goes to $0$ and moreover $k(h)=o(1/h_0)$. The sums converge as a Riemann sum and the errors are then of order $o(1/h_0)$. Therefore we get

\begin{align}
 \sum_{i=1}^k N_{\leq\mathcal{E}}\left(P_D^{i+}\right)  \underset{h_0 \rightarrow 0}{\sim} \frac{1}{2\pi h_0} \int_{r^\star_{max}}^{r^\star_K} dr^\star \sqrt{\mathcal{E} - V \left(r^\star\right)} \cdot  \chi_{V\left(r^\star\right) \leq \mathcal{E}}.
\end{align}
The statement of the Lemma then follows from
\[
 N \left[E-\delta, E+\delta\right] = N_{\leq E+\delta} - N_{\leq E-\delta}
\]
using the previous formula with $\mathcal{E}=E\pm\delta$.
\end{proof}

\subsection{Kerr-AdS}
In the last section we showed that for any fixed given parameters $M>0$, $|a|<l$, $\alpha<\frac{9}{4}$, the eigenvalue problem (\ref{eq:semred2}), $
P_{base} \left(h_0\right) u = \kappa \cdot u$ with Dirichlet conditions,
admits (lots of) eigenvalues $\kappa$ in the range $E-\delta \leq \kappa = h^2_0 \omega^2 < E+\delta$, provided $h_0$ is chosen sufficiently small (i.e.~$\ell$ large). 

As an immediate corollary, we obtain the existence of eigenvalues in the desired range for Schwarzschild-AdS, simply by setting $a=0$, cf.~(\ref{ssp}).
For the Kerr-AdS case, we still need to relate the above result to the full problem, which we recall is the non-linear eigenvalue problem (\ref{eq:semred}) given by
\begin{align*}
P \left(h\right) u = \kappa \cdot u \, \ \ \ \text{with} \ \ \ \ \kappa=\omega^2 h^2 
\end{align*}
and boundary conditions $u \left(r^\star_{max}\right) = 0$ as well as $\int_{r^\star_{max}}^{\pi/2} dr^\star \left[ |u^\prime|^2 + r^2 u^2 \right]<\infty$. 
To achieve this, consider for fixed $|a|<l$, $M>0$, $\alpha<\frac{9}{4}$, the two-parameter family  of linear eigenvalue problems
\begin{align}
Q_{\ell} \left(b^2,\omega^2\right) u = \Lambda \left(b^2,\omega^2\right) u
\end{align}
for the operator
\begin{align}
Q_{\ell} \left(b^2,\omega^2\right) u := -u^{\prime \prime} + \left(V_\sigma\, \mu_{\ell} \left(b^2 a^2 \omega^2\right) u + \left(V_{junk} + V_{mass}-\omega^2 \right)\right) u \, ,
\end{align}
complemented by the above boundary conditions.
Here $b^2 \in \left[0,1\right]$ is a dimensionless parameter and $\omega^2 \in \mathbb{R}^+$.
Our goal is to show that for $b=1$ there exists an $\omega^2$ such that the above problem has a zero eigenvalue and to moreover suitably control the size of this $\omega^2$.

By the results of the previous section, we know that for $b=0$ there exists, for any sufficiently large $\ell$, an $\omega_{0,\ell}^2$ (satisfying $E-\delta \leq \frac{\omega_{0,\ell}^2}{\mu_{\ell} \left(0\right)} \leq E+\delta$) such that $Q_\ell \left(0,\omega^2_{0,\ell}\right)$ admits a zero eigenvalue. Moreover, this eigenvalue is non-degenerate by standard Sturm-Liouville theory. Listing the eigenvalues of $Q_\ell \left(0,\omega^2_{0,\ell}\right)$ in ascending order, let us say that it is the $(n_{\ell})^{th}$ eigenvalue, which is zero. 

The strategy, now, is the following: We will show by an application of the implicit function theorem that for any $b \in \left[0,1\right]$ we can find an $\omega^2_{b,\ell}$ such that the $(n_{\ell})^{th}$ eigenvalue of 
the operator $Q\left(b^2,\omega^2_{b, \ell}\right) $ is zero. As a second step, we will provide a global estimate on the quotient $\frac{\omega_{b, \ell}^2}{\mu_{\ell} \left(b^2 a^2\omega^2_{b,\ell}\right)}$. For this last step an important monotonicity will be exploited.

\begin{lemma}
Fix parameters $|a|<l$, $M>0$ and $\alpha<\frac{9}{4}$. Suppose we are given parameters $b_0 \in \left[0,1\right]$ and $\omega^2_{b_0,\ell} \in \mathbb{R}^+$ such that the $(n_\ell)^{th}$ eigenvalue of $Q_\ell \left(b_0^2,\omega_{b_0,\ell}^2\right)$ is zero. Then, there exists an $\varepsilon>0$ such that 
\begin{enumerate}
\item for any $b^2 \in \left(\max(0,b^2_0-\varepsilon),b^2_0+\varepsilon\right)$ one can find an associated $\omega^2_{b,\ell} \in \mathbb{R}^+$ such that the $(n_\ell)^{th}$ eigenvalue of $Q_\ell \left(b^2,\omega_{b,\ell}^2\right)$ is zero, \label{it:exis}
\item $\omega^2_{b,\ell}$ changes differentiably in $b^2 \in \left(\max(0,b^2_0-\varepsilon),b^2_0+\varepsilon\right)$ and we have the estimate
$$0 \le \frac{d\omega^2_{b,\ell}}{db^2} \le C \omega^2_{b,\ell},$$
for some constant $C> 0$ which is independent of $b_0$ (but may depend on $\ell$). \label{it:der}
\item The $\varepsilon > 0$ can be taken to be independent of $b_0$ (but may depend on $\ell$). \label{it:glo}
\item $\omega^2_{b,\ell}$ satisfies the estimate
$$
c^{-1} \le \frac{\omega^2_{b,\ell}}{\ell(\ell+1)} \le c,
$$
for some $c>0$ depending only on the parameters $a,l,M,\alpha$. \label{it:esw}
\end{enumerate}
\end{lemma}

\begin{proof}
The $n^{th}$ eigenvalue of $Q_\ell \left(b^2,\omega_{b,\ell}^2  \right)$, denoted $\Lambda_n \left(b^2, \omega^2\right)$, moves smoothly in the parameters $b^2$ and $\omega^2$ and we have the formula
\begin{align}
\Lambda_n \left(b^2, \omega^2\right) = \int_{r^\star_{max}}^{\pi/2} dr^\star \psi_n \left(b^2,\omega^2\right) Q_\ell \left(b^2,\omega^2  \right)\psi_n \left(b^2,\omega^2\right) \, 
\end{align}
for the eigenvalue, provided we normalize the associated eigenfunctions $\psi_n \left(b^2,\omega^2\right) $ by $\int_{r^\star_{max}}^{\pi/2} dr^\star | \psi_n \left(b^2,\omega^2\right) |^2=1$. By assumption, $\Lambda_n \left(b_0^2,\omega_{b_0,\ell}^2\right) = 0$.
We compute
\begin{align} 
\frac{\partial \Lambda_n}{\partial \omega^2} \left(b_0^2, \omega_{b_0,\ell}^2\right) &= \int_{r^\star_{max}}^{\pi/2} dr^\star \psi^2_n \left(b_0^2,\omega_{b_0,\ell}^2\right) \cdot  \Big( V_\sigma \cdot \frac{\partial \mu_\ell}{\partial \omega^2} \left(b_0^2, \omega_{b_0,\ell}^2\right)  - 1\Big) \, , \label{mainder1} \\
\frac{\Lambda_n}{\partial b^2}\left(b_0^2, \omega_{b_0,\ell}^2\right) &= \int_{r^\star_{max}}^{\pi/2} dr^\star \psi^2_n \left(b_0^2,\omega_{b_0,\ell}^2\right)  \cdot  \Big( V_\sigma \cdot \frac{\partial \mu_\ell}{\partial b^2} \left(b_0^2, \omega_{b_0,\ell}^2\right) \Big) \, . \label{mainder2}
\end{align}
The angular eigenvalue $\mu_\ell \left(b^2 a^2 \omega^2\right)$ is itself a smooth function of the two parameters $b^2$ and $\omega^2$. We have the formula
\begin{align}
\mu_\ell \left(b^2 a^2 \omega^2\right) = \int_0^\pi d\theta \sin \theta \ \phi_\ell \left(b^2,\omega^2\right) \left[ P_{\theta,\alpha} \left(b^2 a^2 \omega^2\right)  + b^2 a^2 \omega^2 \right] \phi_\ell \left(b^2,\omega^2\right)
\end{align}
for the eigenvalues provided we normalize the associated eigenfunctions $\phi_\ell \left(b^2,\omega^2\right)$ by $\int_0^\pi d\theta \sin \theta \ |\phi_\ell \left(b^2,\omega^2\right)|^2$. Recalling from (\ref{auxfo}) that 
\begin{align}
P_{\theta,\alpha} \left(b^2 a^2 \omega^2\right) f + b^2 a^2 \omega^2 f \nonumber \\
= -\frac{1}{\sin \theta} \partial_\theta \left(\Delta_\theta \sin \theta \partial_\theta f \right) + \frac{\sin^2 \theta}{\Delta_\theta}b^2 a^2 \omega^2 f +\frac{|\alpha|}{l^2} a^2 
\left\{ 
\begin{array}{l} 
\cos^2 \theta f \\ 
\sin^2 \theta f
\end{array} \right.  \, 
\end{align}
to be read with the upper (lower) line in case that $\alpha<0$ ($\alpha \geq 0$), we obtain
\begin{align} \label{eq:dmu}
\frac{\partial \mu_\ell}{\partial \omega^2} \left(b_0^2, \omega_{b_0,\ell}^2\right) = b_0^2 l^2 \cdot T_x \textrm{\ \ \ and \  \ \ } \frac{\partial \mu_\ell}{\partial b^2} \left(b_0^2, \omega_{b_0,\ell}^2\right) = \omega_{b_0,\ell}^2 l^2 \cdot T_x \, ,
\end{align}
where
\begin{align} \label{es:key}
T_x &=  \int_0^\pi d\theta \sin \theta  \frac{\sin^2 \theta}{\Delta_\theta} \frac{a^2}{l^2}  \ |\phi_\ell \left(b_0^2,\omega_{b_0,\ell}^2\right)|^2 \nonumber \\
&\leq \sup_\theta \left[ \frac{\sin^2 \theta}{\Delta_\theta} \frac{a^2}{l^2} \right]\int_0^\pi d\theta \sin \theta  \ |\phi_\ell \left(b_0^2,\omega_{b_0,\ell}^2\right)|^2 <\frac{a^2}{l^2} \,,
\end{align}
where we have used the estimate $\sup_{\theta \in [0, \pi[} \left[\frac{\sin^2 \theta}{\Delta_\theta} \right]\le 1$, which can be easily checked using (\ref{def:deltatheta}).
Going back to (\ref{mainder1}) and (\ref{mainder2}), the implicit function theorem allows us to solve for $\omega^2$ as a (smooth) function of $b^2$ locally near $\left(b_0^2,\omega_{b_0,\ell}^2\right)$, provided that
the right-hand side of (\ref{mainder1}) is non-zero. To achieve this, note that
$$
1- V_\sigma \frac{\partial \mu_\ell}{\partial \omega^2} \ge 1-b_0^2 a^2 V_\sigma \ge 1- a^2 V_\sigma,
$$
using the estimate obtained for $\frac{\partial \mu_\ell}{\partial \omega^2}$ from (\ref{eq:dmu}) and (\ref{es:key}). On the other hand, using (\ref{eq:fs})
$$
a^2 V_\sigma=\frac{a^2}{l^2}+ \frac{\Xi a^2}{r^2+a^2}= \frac{a^2}{l^2} + \Xi - \frac{\Xi r^2}{r^2+a^2}< 1 - c_{M,l,a} \, ,
$$
where $c_{M,l,a}$ is constant depending only the parameters $M,l,a$. It then follows that the right-hand side of (\ref{mainder1}) is bounded away from zero with the lower bound being independent of $b_0$. This concludes the proof of the first item of the lemma.

The implicit function theorem also provides a formula for the derivative of the function $\omega^2\left(b^2\right)$ just obtained, namely
\begin{align} \label{pio}
\frac{d \omega^2_{b,\ell}}{d b^2} \left(b_0\right) = - \frac{\frac{\Lambda_n}{\partial b^2} \left(b_0^2, \omega_{b_0,\ell}^2\right) }{\frac{\Lambda_n}{\partial \omega^2} \left(b_0^2, \omega_{b_0,\ell}^2\right) }\, .
\end{align}
In view of (\ref{eq:dmu}), we immediately obtain an estimate of the form 
\begin{equation} \label{uniex}
0 \le \frac{d \omega^2_{b,\ell}}{d b^2} \left(b_0\right) \le C_{M,l,a} \omega^2_{b_0,\ell} \, ,
\end{equation}
establishing item \ref{it:der}.~of the lemma. As a consequence, we also obtain a uniform bound on $\omega^2_{b,\ell}$ for all $b \in [0,1]$. 

Observe on the other hand that the identity $\Lambda_n\left(\omega^2,b^2\right)=0$ implies (using the Hardy inequality (\ref{Hardy}) as well as $V_\sigma \geq \frac{1}{l^2}$) that 
$\omega^2_{b,\ell} \geq  \frac{1}{2l^2}\mu_\ell \left(b^2\omega^2\right)  \geq \frac{1}{2l^2}\mu_\ell \left(0\right) \geq c_{M,l,a} \ell \left(\ell+1\right)$ for any $b^2 \in \left[0,1\right]$. Hence the quantity $\omega^2_{b,\ell}$ will stay strictly away from zero and
\begin{align} \label{control1}
C_{M,l,a} \geq \frac{\omega^2_{0,\ell}}{\ell \left(\ell+1\right)}\geq C_{M,l,a}^{-1} \frac{\omega^2_{b,\ell}}{\ell \left(\ell+1\right)}\geq c_{M,l,a} \, ,
\end{align}
the first inequality following from the analysis of the \emph{linear} eigenvalue problem (\ref{eq:semred2}) in Section \ref{sec:Weyllin}, and the second from the uniform estimate (\ref{uniex}) on $\frac{d \omega^2_{b,\ell}}{d b^2}$.

Finally, the size of the $\varepsilon$ promised by the implicit function theorem is uniform in $b^2$ by the fact that $\Lambda_n \left(b^2, \omega^2\right)$ is a smooth function on the compact set $\left[0,1\right] \times \left[0, \omega_{\max,\ell}^2\right]$, where $\omega_{\max,\ell}^2$ denotes an upper bound for $\omega_{b^2,\ell}^2$. This, together with the uniform estimate (\ref{es:key}) allows uniform control of the errors arising in the implicit function theorem.
\end{proof}

In view of the fact that $\varepsilon$ is uniform in $b^2_0$, we can apply the implicit function theorem all the way from $b_0=0$ to $b=1$. Note that $\varepsilon$ \emph{can} depend on $\ell$. However, for each \emph{fixed} $\ell$ it only takes finitely many applications of the implicit function theorem to reach $b=1$. 

To gain quantitative global control beyond (\ref{control1}) on the behavior of $\omega^2_{b,\ell}\left(b^2\right)$, let us look at the quotient
\begin{align}
E_n \left(b^2\right) := \frac{\omega_{b,\ell}^2 \left(b^2\right)}{\mu_{\ell} \left(b^2a^2\omega^2_{b,\ell}\left(b^2\right)\right)} \qquad \textrm{with $b^2 \in \left[0,1\right]$}\, ,
\end{align}
which by construction is the $(n_\ell)^{th}$ eigenvalue of the semi-classical operator
\begin{align}
\tilde{Q}_b \left(u\right)  &= E_n \left(b^2\right) u \ \ \ \text{with} \nonumber \\
\tilde{Q}_b \left(u\right)  &= - u^{\prime \prime} \left[\mu_{\ell} \left(b^2a^2\omega_{b,\ell}^2\right)\right]^{-1} + \left[V_{\sigma} +  \left[\mu_{\ell}\left(b^2a^2\omega_{b,\ell}^2\right)\right]^{-1} \left(V_{junk} + V_{mass}\right) \right] u  \, . \nonumber
\end{align}

Recall that $E$ is an energy level such that $E \in (\frac{1}{l^2},V_{\max})$ and that, with the notation just introduced, $E_n(0) \in [E-\delta,E+\delta]$ for some small $\delta > 0$.
\begin{lemma}
For all $\delta'> 0$, there exists a $L$, such that for all $\ell> L$, we have 
\begin{align}
\frac{1}{l^2} \leq E_n \left(1\right) \leq E+\delta + \delta^\prime.
\end{align}
\end{lemma}
\begin{proof}
We first establish the upper bound. Using the Hardy inequality (\ref{Hardy}) together with (\ref{minmaxangular}) (which implies that $\sigma:= \frac{1}{\mu_{\ell} \left(0\right)}-\frac{1}{\mu_{\ell} \left(b^2a^2\omega_{b,\ell}^2\right)} >0$ for any $1 \geq b>0$), we can estimate
\begin{align}
\int_{r^\star_{max}}^{\pi/2} dr^\star \left[ u\left( \tilde{Q}_0 - \tilde{Q}_b\right) u \right] \nonumber \\ 
\geq  \sigma \int_{r^\star_{max}}^{\pi/2} dr^\star |u|^2 \left\{ \frac{1}{4l^2}\frac{\Delta_-}{r^2+a^2} + V_{junk} + V_{mass} \right\} \nonumber \\
\geq \sigma  \int_{r^\star_{max}}^{\pi/2} dr^\star \Bigg( \left[ \left(\frac{9}{4}-\alpha\right) \frac{1}{l^2} \frac{\Delta_-}{r^2+a^2} |u|^2 \right] - C_{M,l,a}|u|^2  \Bigg) \, ,
\end{align}
since $V_{junk}$ is bounded uniformly. In view of $|\sigma|\leq \frac{C_{M,l,a} }{\ell \left(\ell+1\right)}$, we conclude that
\[
\tilde{Q}_0 \geq \tilde{Q}_b -  \frac{C_{M,l,a}}{\ell \left(\ell+1\right)}
\]
holds for any $1 \geq b>0$. Hence, by min-max, we infer in particular that (independently of the parameters $a$, $M$, $\ell$ and $\alpha$)
\begin{align}
E_n\left(1\right) \leq E_n\left(0\right) + \delta^{\prime} \leq E+\delta + \delta^\prime \, ,
\end{align}
where $\delta^\prime$ can be chosen arbitrarily small by choosing $\ell$ sufficiently large.

For the lower bound, we will establish that $\int dr^\star u \left(\tilde{Q}_b \left(u\right) u\right)> \frac{1}{l^2}$ holds for any $u \in H^{2}_{AdS}$ and $\ell$ large. To prove the latter, note first that the Hardy inequality (\ref{Hardy}) reduces the problem to showing that
\begin{align} \label{doit}
 V_\sigma  + h^2 \left[\frac{1}{4} \frac{\Delta_-}{r^2+a^2} + V_{junk} + V_{mass}  \right] > \frac{1}{l^2} \, .
\end{align}
The square bracket is manifestly positive in a region $\left[r^\star_L,\pi/2\right)$ for some large $r^\star_L$ close to $\pi/2$ (depending only on the parameters $M$, $|a|<l$ and $\alpha<\frac{9}{4}$). We fix this $r^\star_L$ and, in view of $V_\sigma \geq \frac{1}{l^2}$, have established (\ref{doit}) in  $\left[r^\star_L,\pi/2\right)$. In $\left[r^\star_{max},r^\star_L\right]$, we have the global estimate
\begin{align}
V_\sigma \geq \frac{1}{l^2} + \frac{Mr}{\left(r^2+a^2\right)^2} \, .
\end{align}
The second term on the right will dominate the term $h^2\cdot \left(V_{junk}\right)$ pointwise in  $\left[r^\star_{max},r^\star_L\right]$ provided $h$ is chosen small depending only on $M$, $l$, $a$, $\alpha$.
\end{proof}

We summarize our results in the following proposition, which can be understood as the analogue of Proposition \ref{prop:eig}.
\begin{proposition}\label{prop:eig2} 
Let $\alpha<\frac{9}{4}$, $M>0$, $|a|<l$ be fixed and $E \in \left(\frac{1}{l^2}, V_{\max}\right)$ be given. Then, there exists an $L > 0$ such that for any $\ell > L$, the following statement holds. There exists an $\omega_\ell^2 \in \mathbb{R}^+$ and a smooth solution $u_{\ell}$ of the axisymmetric reduced equation
\begin{align}
\begin{split}
&-u_\ell^{\prime \prime} + \left[ V_{\sigma} \cdot \mu_\ell \left(a^2\omega_\ell^2\right) + V_{junk} + V_{mass} \right] u_\ell  = \omega_\ell^2 \cdot u _\ell \\
 \textrm{satisfying} & \ \ u_\ell \left(r^\star_{max}\right) = 0 \ \ and \ \ \int_{r^\star}^{\pi/2} \left[ |u^\prime_\ell |^2 + |u_\ell|^2 r^2 \right]dr^\star < \infty \, .
 \end{split}
\end{align}
Moreover, the $\omega^2_\ell$ satisfy the uniform estimates
\begin{align}
\frac{1}{l^2} <  \frac{\omega_\ell^2}{\mu_\ell \left(a^2\omega_\ell^2\right)} \leq E+\frac{V_{max}-E}{2} \ \ \  \textrm{and} \ \ \  c_{M,l,a} \leq \frac{\omega_\ell^2}{\ell\left(\ell+1\right)} \leq C_{M,l,a} \, . 
\end{align}
\end{proposition}

\section{Agmon estimates} \label{se:ae}
In this section, we recall the so-called Agmon estimates. These are (well known) exponential decay estimates for eigenfunctions for Schr\"odinger type operators, in the so-called \emph{forbidden} regions. 

\subsection{Energy inequalities}

The Agmon estimates will rely on the following identity

\begin{lemma}[Energy identity for conjugated operator] \label{lem:eico}
Let $r^\star_1 > r^\star_0$. Let $h > 0$ and let $W$, $\phi$ be smooth real valued functions on $[r^\star_0,r^\star_1]$. For all smooth functions $u$ defined on $[r^\star_0,r^\star_1]$, we have the identity
\begin{eqnarray}
&&\int^{r^\star_1}_{r^\star_0}\left( \left|\frac{d}{dr^\star}\left(e^{ \frac{\phi}{h}} u \right) \right|^2+h^{-2}\left(W-\left(\frac{d\phi}{d r^\star}\right)^2\right) e^{2 \frac{\phi}{h}}|u|^2 \right)dr^\star= \nonumber \\
&&\hbox{}\int^{r^\star_1}_{r^\star_0}\left(-\frac{d^2\overline{u}}{dr^\star\phantom{}^2}+h^{-2}W\overline{u}\right) u e^{2\frac{\phi}{h}}dr^\star+\int^{x_1}_{x_0} h^{-1} \frac{d\phi}{d r^\star} e^{2\frac{\phi}{h}} 2i \Im \left(\overline{u} \frac{d u}{d r^\star}  \right)d r^\star\nonumber \\
&&\hbox{}+\left(e^{2\frac{\phi}{h}}\frac{d\overline{u}}{dr^\star}u \right)(r^\star_1)-\left(e^{2\frac{\phi}{h}}\frac{d\overline{u}}{dr^\star}u \right)(r^\star_0)\,.\nonumber 
\end{eqnarray}
In particular, if $u$ is real valued and vanishes at $r^\star_0$ and $r^\star_1$, then,
\begin{eqnarray}
&&\int^{r^\star_1}_{r^\star_0}\left( \left|\frac{d}{dr^\star}\left(e^{ \frac{\phi}{h}} u \right) \right|^2+h^{-2}\left(W-\left(\frac{d\phi}{d r^\star}\right)^2\right) e^{2 \frac{\phi}{h}}|u|^2 \right)dr^\star= \nonumber \\
&&\hbox{}\int^{r^\star_1}_{r^\star_0}\left(-\frac{d^2u}{dr^\star\phantom{}^2}+h^{-2}Wu\right) u e^{2\frac{\phi}{h}}dr^\star\,. \nonumber 
\end{eqnarray}
Moreover, the same idendity holds if $W$ is not assumed smooth on $[r_0^\star,r_1^\star]$ but only such that $W |u|^2 \in L^1(r_0^\star,r_1^\star)$. By density, we may also replace smoothness of $u$ and $\phi$ by $u \in H^1_{0}[r^\star_0,r^\star_1]$ and $\phi$ is a Lipschitz function.
\end{lemma}
\begin{proof}
This follows easily from the computations
\begin{eqnarray}
\int^{r^\star_1}_{r^\star_0} \left( -\frac{d^2}{dr^\star\phantom{}^2} +h^{-2} W\right)\overline{(e^{\frac{\phi}{h}} u)} u e^{\frac{\phi}{h}} dr^\star =\int^{r^\star_1}_{r^\star_0} \left|\frac{d}{dr^\star} e^{\frac{\phi}{h}} u \right|^2+h^{-2} W e^{2\frac{\phi}{h}}|u|^2 dr^\star \nonumber \\
-\left(\overline{\frac{d}{dr^\star}( e^{\frac{\phi}{h}}u)}ue^{\frac{\phi}{h}}\right)(r^\star_1)+\overline{\frac{d}{dr^\star}( e^{\frac{\phi}{h}}u)}ue^{\frac{\phi}{h}}(r^\star_0)  \nonumber
\end{eqnarray}
and 
\begin{eqnarray}
& \int^{r^\star_1}_{r^\star_0}\left(  -\frac{d^2}{dr^\star\phantom{}^2}\right)\overline{(e^{\frac{\phi}{h}} u)} u e^{\frac{\phi}{h}} dr^\star=
\int^{r^\star_1}_{r^\star_0} -\frac{d}{dr^\star}\left( h^{-1} \frac{d\phi}{d r^\star} \overline{u}e^{\frac{\phi}{h}}+\frac{d \overline{ u} }{d r^\star}e^{ \frac{\phi}{h}}\right)
u e^{ \frac{\phi}{h}}dr^\star \nonumber \\
&=\int^{r^\star_1}_{r^\star_0}\left[-h^{-1}  \frac{d\phi}{d r^\star} \frac{d \overline{ u} }{d r^\star}ue^{2 \frac{\phi}{h}}-\frac{d^2\overline{u}}{dr^\star\phantom{}^2 }u e^{2 \frac{\phi}{h}}\right]dr^\star \nonumber\\
& \phantom{XXXXXX} +\int^{r^\star_1}_{r^\star_0}\left[ h^{-1} \frac{d\phi}{d r^\star} \frac{d u}{d r^\star} \overline{u}e^{2\frac{\phi}{h}}+h^{-2}  \left( \frac{d\phi}{d r^\star} \right)^2 \overline{u} u e^{2\frac{\phi}{h}}\right]dr^\star \nonumber \\
&\phantom{XXXXXX} -h^{-1}\left( \frac{d\phi}{d r^\star} |u|^2 e^{2\frac{\phi}{h}}\right)(r^\star_1)+ h^{-1}  \left(\frac{d\phi}{d r^\star} |u|^2 e^{2 \frac{\phi}{h}}\right)(r^\star_0)\nonumber \\
&=\int^{r^\star_1}_{r^\star_0} \left( h^{-2} \left(\frac{d\phi}{d r^\star} \right)^2 |u|^2 e^{2\frac{\phi}{h}}-\frac{d^2\overline{u}}{dr^\star\phantom{}^2 } u e^{2\frac{\phi}{h}} \right)dr^\star\nonumber \\
&+\int^{r^\star_1}_{r^\star_0}h^{-1} \frac{d\phi}{d r^\star} e^{2\frac{\phi}{h}} 2i \Im (\overline{u} u_{r^\star} )dr^\star\nonumber \\
&\phantom{XXXXXX} -h^{-1}\left( \frac{d\phi}{d r^\star} |u|^2 e^{2\frac{\phi}{h}}\right)(r^\star_1)+h^{-1} \left(\frac{d\phi}{d r^\star} |u|^2 e^{2\frac{\phi}{h}} \right)(r^\star_0)\,.\nonumber 
\end{eqnarray}
\end{proof}







\subsection{The Agmon distance}
We will rely on the Agmon distance to establish our exponential decay estimates.\footnote{The Agmon distance is actually typically used to obtain \emph{optimal} exponential decay estimates, see for instance \cite{hs:smss}. For the main purpose of this paper (the construction of quasimodes), we could have used smooth cut-off constructions to prove slightly weaker exponential decay estimates. However, the Agmon distance (despite leading only to Lipschitz cut-offs) has a nice interpretation which is why we choose to use it here.}
Given any energy level $\mathcal{E} > 0$, and a potential $V=V(r^\star)$ (which may also depends on an parameter $h$), we define the Agmon distance $d$ between $r^\star_1$ and $r^\star_2$ as
$$
d=d_{(V-\mathcal{E})_+}(r^\star_1,r^\star_2)=\left| \int_{r_1^\star}^{r_2^\star}\chi_{\{V \ge \mathcal{E} \} }(r^\star) \left( V(r^\star)-\mathcal{E} \right)^{1/2}dr^\star\right|,
$$
where $\chi_{\{V \ge \mathcal{E} \}}$ is the characteristic function of the set of $r^\star$ satisfying $V(r^\star) \ge  \mathcal{E}$. In other words, $d$ is the distance associated to the Agmon metric $(V-\mathcal{E})_+d r^\star\phantom{}^2$, where $f_+=\min(0,f)$ for any function $f$.

It is easily checked that $d$ satisfies the triangular inequality and that 
$$
|\nabla_{r^\star} d(r^\star,r^\star_2) |^2 \le (V-\mathcal{E})_+(r^\star) \, .
$$
The distance to a set can also be defined as usual. In particular, we define 
$$
d_{\mathcal{E}}(r^\star):= \inf_{r^\star_0 \in \{ \mathcal{E} \ge V\}} d(r^\star,r^\star_0) \, ,
$$
which measures the distance to the classical region. We have again 
$$| \nabla_{r^\star} d_{\mathcal{E}}(r^\star) |^2 \le (V-\mathcal{E})_+(r^\star).$$
For a given small $\epsilon \in (0,1)$ we define the two $r^\star$-regions 
$$
\Omega^+_\epsilon=\Omega^+_\epsilon(\mathcal{E}):=\{r^\star \ | \ V(r^\star) > \mathcal{E}+\epsilon \}
$$
and its complement
$$
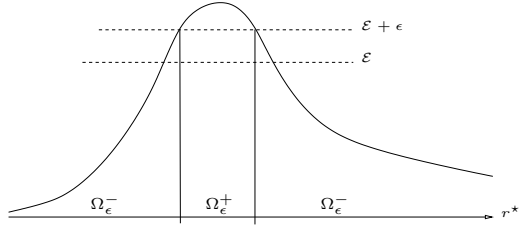
\Omega^-_\epsilon=\Omega^-_\epsilon(\mathcal{E}):=\{r^\star \ | \ V(r^\star) \le\mathcal{E} +\epsilon \}.
$$
\[
\begin{picture}(0,0)%
\includegraphics{agmon.pstex}%
\end{picture}%
\setlength{\unitlength}{1342sp}%
\begingroup\makeatletter\ifx\SetFigFont\undefined%
\gdef\SetFigFont#1#2#3#4#5{%
  \reset@font\fontsize{#1}{#2pt}%
  \fontfamily{#3}\fontseries{#4}\fontshape{#5}%
  \selectfont}%
\fi\endgroup%
\begin{picture}(9027,4110)(1189,-5159)
\put(7651,-2161){\makebox(0,0)[lb]{\smash{{\SetFigFont{5}{6.0}{\rmdefault}{\mddefault}{\updefault}{\color[rgb]{0,0,0}$\mathcal{E}$}%
}}}}
\put(7651,-1561){\makebox(0,0)[lb]{\smash{{\SetFigFont{5}{6.0}{\rmdefault}{\mddefault}{\updefault}{\color[rgb]{0,0,0}$\mathcal{E}+\epsilon$}%
}}}}
\put(10201,-5011){\makebox(0,0)[lb]{\smash{{\SetFigFont{5}{6.0}{\rmdefault}{\mddefault}{\updefault}{\color[rgb]{0,0,0}$r^\star$}%
}}}}
\put(2701,-4861){\makebox(0,0)[lb]{\smash{{\SetFigFont{6}{7.2}{\rmdefault}{\mddefault}{\updefault}{\color[rgb]{0,0,0}$\Omega_\epsilon^-$}%
}}}}
\put(4801,-4861){\makebox(0,0)[lb]{\smash{{\SetFigFont{6}{7.2}{\rmdefault}{\mddefault}{\updefault}{\color[rgb]{0,0,0}$\Omega_\epsilon^+$}%
}}}}
\put(6901,-4861){\makebox(0,0)[lb]{\smash{{\SetFigFont{6}{7.2}{\rmdefault}{\mddefault}{\updefault}{\color[rgb]{0,0,0}$\Omega_\epsilon^-$}%
}}}}
\end{picture}%

\]
\subsection{The main estimate}
We would like to apply Lemma \ref{lem:eico} between $r^\star_{\max}$ and $\pi/2$ for $u$ a solution to the eigenvalue problem (\ref{eq:semred}) and for suitable $\phi$. 
\begin{lemma}
Let $u$ be a solution to the eigenvalue problem (\ref{eq:semred}), i.e.~$\kappa \cdot u=P(h)u$  for some $\kappa=h^2\omega^2$. Define for any $\epsilon 
\in \left(0,1\right)$
\begin{align}\label{phidef}
\phi_{\kappa, \epsilon}:=(1-\epsilon) d_\kappa.
\end{align}
Then, for all $\epsilon$ sufficiently small, $u$ satisfies
\begin{align} \label{maingp}
\int^{\pi/2}_{r^\star_{\max}}h^2 \left|\frac{d}{dr^\star} e^{\frac{\phi_{\kappa,\epsilon}}{h}} u \right|^2 dx \ + \ &\epsilon^2 \int_{\Omega^+_\epsilon} e^{2 \frac{\phi_{\kappa,\epsilon}}{h}}|u|^2 dr^\star \nonumber \\
& \leq D(\kappa+\epsilon)e^{2 a(\epsilon)/h}|| u ||^2_{L^2(r^\star_{\max},\pi/2)} \, ,
\end{align}
where $a(\epsilon)=\sup_{\Omega^-_\epsilon}d_\kappa$ and $D > 0$ is a constant depending only on the parameters $a,M,l$ and $\alpha$.
\end{lemma}
\begin{remark}
Note 
$$a_\kappa(\epsilon)=\sup_{\Omega^-_\epsilon}d_\kappa \rightarrow 0$$ 
as $\epsilon \rightarrow 0$, uniformly in $h$ (and $\kappa$) for $h$ sufficiently small. In view of the exponential weight in the second term on the left, the estimate (\ref{maingp}) quantifies that $u$ is exponentially small in the forbidden region, provided we can show a uniform lower bound for $\phi_{\kappa,\epsilon}$ in a suitable subset of $\Omega_\epsilon^+$. This will be achieved in Lemma \ref{lem:fpp}.
\end{remark}

\begin{proof}
Applying Lemma \ref{lem:eico} between $r^\star_{\max}$ and $\pi/2$, we get
\begin{eqnarray} \label{ineq:sp}
&&\int^{\pi/2}_{r^\star_{\max}} h^2\left|\frac{d}{d r^\star} e^{ \frac{\phi}{h}} u \right|^2 dr^\star + \int_{\Omega^+_\epsilon}\left(V-\kappa-\Big|\frac{d \phi}{dr^\star}\Big|^2\right) e^{2 \frac{\phi}{h}}|u|^2 dr^\star=\nonumber \\
&&\hbox{}\,\,\int_{\Omega^-_\epsilon}\left(\kappa-V+\Big|\frac{d \phi}{dr^\star}\Big|^2\right) e^{2 \frac{\phi}{h}}|u|^2 dr^\star.
\end{eqnarray}
In view of our choice $\phi=\phi_{\kappa,\epsilon}$, we have in $\Omega^+_\epsilon$ the estimate
\begin{eqnarray} \label{intest}
V-\kappa-\left|\frac{d \phi_{\kappa,\epsilon}}{dr^\star}\right|^2&\ge& \left(1-(1-\epsilon)^2 \right) (V-\kappa) \ge \epsilon^2,
\end{eqnarray}
for $\epsilon$ sufficiently small, which we will use to estimate the left hand side of \eqref{ineq:sp}.

For the right-hand side of \eqref{ineq:sp}, we note that if $V \ge 0$ (which occurs if $\alpha \le 2$), then we immediately obtain
\begin{align} \label{aux2}
\int_{\Omega^-_\epsilon}\left[\kappa-V+\Big|\frac{d \phi_{\kappa,\epsilon}}{dr^\star}\Big|^2\right] e^{2 \frac{\phi_{\kappa,\epsilon}}{h}}|u|^2 dr^\star \le (\kappa+\epsilon)e^{2 a(\epsilon)/h}|| u ||^2_{L^2(r^\star_{\max},\pi/2)} \, ,
\end{align}
so that combining (\ref{intest}) and (\ref{aux2}) yields (\ref{maingp}).

To obtain (\ref{maingp}) also in the case $\alpha > 2$ (for which we have $V(r^\star)\rightarrow -\infty$ as $r^\star \rightarrow \pi/2$), we need once again to appeal to a Hardy-type inequality to absorb the error by the derivative term on the left-hand side of \eqref{ineq:sp}. 

This we do as follows. Recall that $V=V_\sigma+h^2( V_{junk} + V_{mass})$, and that the unbounded term is $h^2 V_{mass}=h^2 \frac{2-\alpha}{l^2} \frac{ \Delta_- r^2}{\left( r^2+a^2\right)^2}<0$ for $\alpha>2$.

Note that $\frac{ \Delta_- r^2}{\left( r^2+a^2\right)^2}=\frac{ \Delta_-}{\left( r^2+a^2\right)}-\frac{ \Delta_- a^2}{\left( r^2+a^2\right)^2}$. The second term is bounded (and in fact will contribute with the right sign if $\alpha \ge 2)$ so its contribution can be treated as before. Thus, we only need to estimate $\int_{\Omega^-_\epsilon} \frac{ \Delta_-}{\left( r^2+a^2\right)} e^{2 \frac{\phi_{\kappa,\epsilon}}{h}}|u|^2 dr^\star$.

By Lemma 7.2  of \cite{gs:dpkads} (cf.~(\ref{Hardy})) we have for any function $v$ in $H^1_0(r^\star_{\max},\pi/2)$ 
%
%
\begin{eqnarray}
\int^{\pi/2}_{R^\star} \frac{ \Delta_-}{\left( r^2+a^2\right)} |v|^2dr^\star \le 4 l^2 \int_{R^\star}^{\pi/2} \Big|\frac{dv}{dr^*}\Big|^2 dr^* \ \ \ \textrm{for any $R^\star \geq r^\star_{max}$} \, .
\end{eqnarray}
Applying the above Hardy inequality to $v=e^{\phi_{\kappa,\epsilon}/h}u$, we obtain that there exists a uniform constant $C>0$ such that
\begin{align} 
\int^{\pi/2}_{r^\star_{\max}} h^2 \left(9/4-\alpha \right)\left|\frac{d}{d r^\star} e^{ \frac{\phi_{\kappa,\epsilon}}{h}} u \right|^2 dr^\star +\int_{\Omega^+_\epsilon}\left(V-\kappa-\Big|\frac{d \phi_{\kappa,\epsilon}}{dr^\star} \Big|^2\right) e^{2 \frac{\phi_{\kappa,\epsilon}}{h}}|u|^2 dr^\star
\nonumber \\
\leq C (\kappa+\epsilon)e^{2 a(\epsilon)/h}|| u ||^2_{L^2(r^\star_{\max},\pi/2)} \, ,
\nonumber
\end{align}
i.e.~there exists a constant $D>0$ (which degenerates as $\alpha \rightarrow 9/4$) such that 
\begin{eqnarray} \label{ineq:aph}
&&\int^{\pi/2}_{r^\star_{\max}} h^2 \left|\frac{d}{d r^\star} e^{ \frac{\phi_{\kappa,\epsilon}}{h}} u \right|^2 dr^\star +\int_{\Omega^+_\epsilon}\left(V_\sigma-\kappa-\Big|\frac{d \phi_{\kappa,\epsilon}}{dr^\star} \Big|^2\right) e^{2 \frac{\phi_{\kappa,\epsilon}}{h}}|u|^2 dr^\star
\nonumber \\
&\le&\hbox{}\,\, D (\kappa+\epsilon)e^{2 a(\epsilon)/h}|| u ||^2_{L^2(r^\star_{\max},\pi/2)}.
\end{eqnarray}
This estimate, when combined with (\ref{intest}), yields again \eqref{maingp} from (\ref{ineq:sp}).
\end{proof}

\subsection{Application of the main estimate}
Before we can exploit \eqref{maingp}, we need the following Lemma, which quantifies the size of the forbidden region for a given energy level.
\begin{lemma}\label{lem:fpp} 
Let $E \in \left(\frac{1}{l^2},V_{\max}\right)$ and suppose that $\kappa \in (\frac{1}{l^2},E+\delta]$ for some $\delta>0$ such that $E+\delta < V_{\max}$. Then there exists a $\delta' > 0$ and a $C>0$, both constants being independent of $h$, such that $V_\sigma-\kappa> 2C$, in $[r_{\max},r_{\max}+\delta']$, for all $\kappa \in [E-\delta,E+\delta]$. 
\end{lemma}
\begin{proof}
This is a simple consequence of the continuity of $V_\sigma$ at $r^\star_{\max}$.
\end{proof}
In view of the full potential being $V = V_\sigma + h^2 \left(V_{junk} + V_{mass}\right)$ we also obtain
\begin{corollary} 
For $h$ sufficiently small (depending only on $M$, $l$ and $a$) we have $V-\kappa > C$ in $[r_{\max},r_{\max}+\delta']$ for all $\kappa \in [E-\delta,E+\delta]$ with both $\delta^\prime$ and $C$ depending only on $M$, $l$ and $a$.
\end{corollary}
With $E \in \left(\frac{1}{l^2},V_{\max}\right)$ given, we now fix $\delta'>0$ and $C>0$ as promised by Lemma \ref{lem:fpp}. This implies that $\phi_{\kappa,\epsilon} \geq c_{M,l,a}$ in $\left[r_{max},r_{max}+\delta^\prime\right]$ uniformly in $\epsilon$ (the constant $c_{M,l,a}$ being of size $C\cdot\delta^\prime$). Next we fix $\epsilon>0$ sufficiently small so that in particular
$a(\epsilon) \le c_{M,l,a}/2$. We finally conclude from \eqref{maingp} that there exists a $\widetilde{C}>0$ (independent of $h$) such that
\begin{align} \label{no1}
\int_{r^\star_{\max}}^{r^\star_{\max}+\delta'} |u|^2 dr^\star\le \widetilde{C} e^{-\widetilde{C}/h}|| u ||^2_{L^2(r^\star_{\max},\pi/2)} \, .
\end{align}
Turning to the derivative term on the left of \eqref{ineq:aph}, we also have
\begin{align*}
\int_{r^\star_{\max}}^{r^\star_{\max}+\delta'} h^2 e^{2 \phi_{\kappa,\epsilon}/h}\bigg( \frac{1}{h^2}\left|\frac{d \phi_{\kappa,\epsilon}}{dr^\star}\right|^2|u|^2 +\frac{2}{h} \frac{d \phi_{\kappa,\epsilon}}{dr^\star} u& \frac{d u}{dr^\star}+\left|\frac{d u}{dr^\star}\right|^2 \bigg)dr^\star\lesssim \\
&\hbox{} e^{2a(\epsilon)/h}|| u ||^2_{L^2(r^\star_{\max},\pi/2)} \, .
\end{align*}
The $|u|^2$ term in the above integral can be ignored since it has the right sign. The crossterm can be absorbed using (\ref{no1}) and $\frac{1}{2}$ of the derivative-term. Therefore,
$$
\int_{r^\star_{\max}}^{r^\star_{\max}+\delta'} \left|\frac{d u}{dr^\star}\right|^2 dr^\star\le \widetilde{C} h^{-2} e^{-\widetilde{C}/h}|| u ||^2_{L^2(r^\star_{\max},\pi/2)} \, .
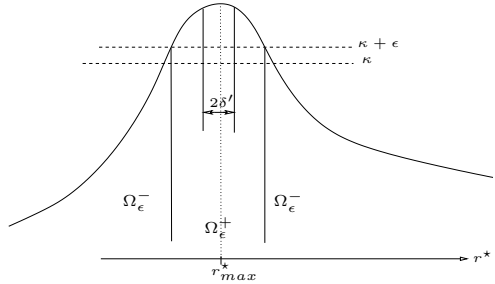
$$
\[
\begin{picture}(0,0)%
\includegraphics{agmon2.pstex}%
\end{picture}%
\setlength{\unitlength}{1342sp}%
\begingroup\makeatletter\ifx\SetFigFont\undefined%
\gdef\SetFigFont#1#2#3#4#5{%
  \reset@font\fontsize{#1}{#2pt}%
  \fontfamily{#3}\fontseries{#4}\fontshape{#5}%
  \selectfont}%
\fi\endgroup%
\begin{picture}(8874,5071)(1189,-6120)
\put(4785,-5285){\makebox(0,0)[lb]{\smash{{\SetFigFont{6}{7.2}{\rmdefault}{\mddefault}{\updefault}{\color[rgb]{0,0,0}$\Omega_\epsilon^+$}%
}}}}
\put(9686,-5823){\makebox(0,0)[lb]{\smash{{\SetFigFont{5}{6.0}{\rmdefault}{\mddefault}{\updefault}{\color[rgb]{0,0,0}$r^\star$}%
}}}}
\put(4873,-2971){\makebox(0,0)[lb]{\smash{{\SetFigFont{5}{6.0}{\rmdefault}{\mddefault}{\updefault}{\color[rgb]{0,0,0}$2\delta^\prime$}%
}}}}
\put(4896,-6051){\makebox(0,0)[lb]{\smash{{\SetFigFont{5}{6.0}{\rmdefault}{\mddefault}{\updefault}{\color[rgb]{0,0,0}$r^\star_{max}$}%
}}}}
\put(7651,-2161){\makebox(0,0)[lb]{\smash{{\SetFigFont{5}{6.0}{\rmdefault}{\mddefault}{\updefault}{\color[rgb]{0,0,0}$\kappa$}%
}}}}
\put(7576,-1861){\makebox(0,0)[lb]{\smash{{\SetFigFont{5}{6.0}{\rmdefault}{\mddefault}{\updefault}{\color[rgb]{0,0,0}$\kappa+\epsilon$}%
}}}}
\put(3287,-4786){\makebox(0,0)[lb]{\smash{{\SetFigFont{6}{7.2}{\rmdefault}{\mddefault}{\updefault}{\color[rgb]{0,0,0}$\Omega_\epsilon^-$}%
}}}}
\put(6047,-4796){\makebox(0,0)[lb]{\smash{{\SetFigFont{6}{7.2}{\rmdefault}{\mddefault}{\updefault}{\color[rgb]{0,0,0}$\Omega_\epsilon^-$}%
}}}}
\end{picture}%

\]
Summarizing these decay estimates, we have proven:
\begin{lemma} \label{lem:dec}
Let $E \in (\frac{1}{l^2},V_{\max})$ be fixed and let $\delta$ be sufficiently small so that $[E-\delta,E+\delta] \subset (\frac{1}{l^2},V_{\max})$. Then, there exists constants $D,\delta'>0$ depending only on the parameters $M$, $l$, $a$ and $\alpha$, such that the sequence of eigenfunctions $\left[u_\ell\right]_{\ell \geq L}^\infty$ arising from Proposition \ref{prop:eig2} satisfies the estimate
$$
\int_{r^\star_{\max}}^{r^\star_{\max}+\delta'}\left(  \left|\frac{d u_\ell}{dr^\star}\right|^2+ |u_\ell|^2 \right) dr^\star\le D e^{-D/h}|| u_\ell ||^2_{L^2(r^\star_{\max},\pi/2)} \, ,
$$
where $h=\left(\mu_\ell\left(a^2 \omega_\ell^2\right)\right)^{-1/2}$ and $\omega_\ell$ are as in Proposition \ref{prop:eig2}. 
\end{lemma}
We remark that by reusing once again the equation, we can obtain such an exponential decay estimates on all higher order derivatives, with the constants in the above lemma depending on the order of commutation.
\section{The construction of quasimodes} \label{se:cq}
By now we have established the existence (Proposition \ref{prop:eig2}) of a sequence of functions $[u_\ell]$ such that for each $\ell$ the corresponding $u_\ell$ solves
$$
 \omega_\ell^2 h^2 u_\ell =P(h) u_\ell,
$$
where $h=\left(\mu_{\ell}(a^2\omega^2_\ell))\right)^{-1/2} \rightarrow 0$ as $\ell \rightarrow \infty$, and such that these $u_\ell$ obey the estimate of Lemma \ref{lem:dec} with some constants $D,\delta'>0$ independent of $h$ (or equivalently $\ell$). \newline

Let now $\chi$ be a smooth function such that $\chi=1$ on $[r^\star_{\max}+\delta',\pi/2]$ and $\chi=0$ on $(-\infty,r^\star_{\max}]$. We then define $\psi_\ell(t,r,\theta,\tilde{\phi})$ as 
\begin{align} \label{qmode}
\psi_\ell(t,r,\theta,\tilde{\phi})=e^{i\omega_\ell t} \chi(r^\star(r)) (r^2+a^2)^{-1/2} u_\ell(r^\star(r)) S_{\ell 0}(\theta).\end{align}

\begin{remark}
As defined above, the $\psi_\ell$ are complex functions, but of course, we could have worked below with $Re ( \psi_\ell)$ or $Im  (\psi_\ell)$.
\end{remark}

The next Lemma shows that the $\psi_\ell$ satisfy the Klein-Gordon equation up to an exponentially small error:
\begin{lemma} \label{lem:Flest}
For each $\ell$ and each $k\ge 0$, $\psi_\ell \in CH^k_{AdS}$. Moreover, there exists $L > 0$ such that we have the following estimates. For all $k \ge 0$, there exists a $C_k>0$ such that for all $\ell \ge L$, for all $t \ge t_0$, 
$$
||\square_g \psi_\ell+\frac{\alpha}{l^2}\psi_\ell ||_{H^k_{Ads}(\Sigma_t)} \le C_k e^{-C_k \ell} || \psi_\ell ||_{H^0_{AdS}\left(\Sigma_{t_0}\right)}.
$$
Finally, all the $H^k_{AdS}$ norms of $\psi_\ell$ and of its time derivatives on each $\Sigma_t$ are constant in $t$.
\end{lemma}

\begin{proof}
Note that by standard elliptic estimates, any $u_\ell$ is smooth on $(r^\star_{\max},\pi/2)$. Thus, as far as the regularity of $\psi_\ell$ is concerned, it is sufficient to check that $\psi_\ell$ and its derivative decay sufficiently fast near $r=\infty$, which is easy and therefore omitted.

Moreover, in view of our construction, we have $\square_g \psi_\ell =0$ in $[r_+,r_{\max}] \cup [r(r^\star_{\max}+\delta'),\infty]$. Hence, the error is supported in a bounded strip in which we have the following naive estimate: For all $(t,r,\theta,\phi)$ with $r \in [r_{\max},r_{\max}+\delta']$,
$$
|\square_g \psi_\ell + \frac{\alpha}{l^2}\psi_\ell| \lesssim \left( \omega^2 |u_\ell| + |u_\ell''|+|u'_\ell|+h^{-2}|u_\ell|\right) S_{\ell 0}(\theta) ,
$$
which gives the required estimate for $k=0$ after integration, using the Agmon estimates of the previous section and the equation satisfied by $u_\ell$ in order to estimate $u''_\ell$. 
For higher $k$, it suffices to commute the equation and to use the equation for $u_\ell$ every time two radial derivatives occur, or the equation for $S_{\ell 0}$ every time angular derivatives occur.
\end{proof}

Note that we finally proved Theorem \ref{th:main}. Indeed, the $\psi_\ell$ are of the form claimed in \ref{prop1}.~by construction of (\ref{qmode}). The estimate on the $\omega_\ell$ in \ref{prop2}.~was obtained as part of Proposition \ref{prop:eig2}. The error-estimate \ref{prop3}.~ist the statement of Lemma \ref{lem:Flest}, while the localization properties \ref{prop4}.~and \ref{prop5}.~are obvious from (\ref{qmode}) itself.

\section{Proof of Corollary \ref{cor}}
In this section, we prove Corollary \ref{cor}. Given the quasimodes, the proof is standard, but we include it for the paper to be self contained. 

Let us therefore fix a Kerr-AdS spacetime such that the assumption of Corollary \ref{cor} are satisfied and also a Klein-Gordon mass $\alpha < 9/4$. For convenience, we set $t^\star_0=0$. Recall also that $t^\star=t$ in $r \ge r_{\max}$.

 We shall consider solutions $\psi$ to homogeneous and inhomogeneous Klein-Gordon equations with initial data $\psi|_{\Sigma_t}$ and $\partial_t \psi |_{\Sigma_t}$ given on slices of constant $t$. We shall avoid completely issues regarging the facts that $\partial_t$ is not always timelike and that the coordinate $t$ breaks down at the horizon by considering only axisymmetric data which is compactly supported away from the horizon.

Thus, given any $t,s \in \mathbb{R}$ and given any smooth, axisymmetric initial data set $w=\left(\overline{\psi},\overline{T\psi}\right)$, whose support is bounded away from the horizon and which decays sufficiently fast near infinity, we will denote by $P(t,s)w$ the unique solution at time $t$ of the homogenous problem
\begin{eqnarray}
\left(\square_g + \frac{\alpha}{l^2} \right)\psi&=&0, \nonumber \\ 
\psi|_{\Sigma_s}&=&\overline{\psi}, \nonumber \\ 
\frac{\partial \psi}{\partial t}\phantom{}\Big|_{\Sigma_s}&=&\overline{T\psi}. \nonumber
\end{eqnarray}

Given a smooth axisymmetric function $F$ defined on $\mathcal{R}$, compactly supported in $r$ away from the horizon and infinity,  we can consider the inhomogeneous problem
\begin{eqnarray}
\left(\square_g + \frac{\alpha}{l^2} \right)\psi&=& F \nonumber \\ 
\psi|_{\Sigma_0}&=&\overline{\psi}, \nonumber \\ 
\frac{\partial \psi}{\partial t}\phantom{}\Big|_{\Sigma_0}&=&\overline{T\psi}. \nonumber
\end{eqnarray}
For regular data as above, this problem is well-posed in $CH^2_{AdS}$ and we shall denote its solution by $\psi_F(t)$, suppressing the dependence on $r$ and the angular variables. 
If the data is axisymmetric, then $\psi_F$ will be axisymmetric and writing $v(s)=(0, F(s)(g^{tt})^{-1})$, $\psi_F(t)$ is given by the Duhamel formula
$$
\psi_F(t)=P(t,0)w+\int_0^t P(t,s)v(s)ds.
$$

We now consider the family of $\psi_\ell$ given by Theorem \ref{th:main}. For each $\ell$, $\psi_\ell$ provides an initial data set $w_\ell=\left(\psi_\ell(t=0),\frac{\partial \psi_\ell}{\partial t}(t=0) \right)$ for \eqref{mwe} on the slice $t=0$. 
Moreover,  $\psi_\ell$ satisfy the inhomogeneous Klein-Gordon equation
\begin{eqnarray}
 \left(\square_g + \frac{\alpha}{l^2} \right)\psi_\ell&=&F_\ell \, , \nonumber
\end{eqnarray}
for some $F_\ell$ satisfying $||F_\ell||_{H^{k,-2}_{AdS}} \le C_k e^{-C_k\ell} || \psi_\ell ||_{H^0_{AdS}}$.

Let $\widetilde{\psi}_\ell$ denote the solution of the \emph{homogeneous} problem associated with the \emph{same} initial data $w_\ell$, i.e. $ \widetilde{\psi}_\ell=P(t,0)w_\ell$. From Duhamel's formula, we then get
\begin{eqnarray} \label{dfg}
||\psi_\ell- \widetilde{\psi}_\ell||_{H^1_{AdS}( \Sigma_t \cap \{ r \ge r_{\max}\})} &\le& t \sup_{ s \in [0,t]} ||P(t,s)(0,F_\ell)(s) ||_{H^1_{AdS}( \Sigma_t \cap \{ r \ge r_{\max}\})}  \nonumber \\
&\le& t C || F_\ell ||_{H^{0,-2}_{AdS}\left(\Sigma_0\right)} \nonumber \\
&\le& t C e^{-C \ell} || \psi_\ell ||_{H^0_{AdS}( \Sigma_0 \cap \{ r \ge r_{\max}\})}, \nonumber \\
&\le& t C e^{-C \ell} || \psi_\ell ||_{H^1_{AdS}( \Sigma_0 \cap \{ r \ge r_{\max}\})} \, ,
\end{eqnarray}
where we have used the boundedness statement of Theorem \ref{theo:pre} to bound \\$||P(t,s)(0,F_\ell)(s) ||_{H^0_{AdS}( \Sigma_t \cap \{ r \ge r_{\max}\})}$ in terms of the data, as well as Lemma \ref{lem:Flest}.
In particular, since the norms of $\psi_{\ell}$ are time invariant, for any $t \le  \frac{e^{C \ell}}{2C}$, the reverse triangle inequality and (\ref{dfg}) yield
\begin{eqnarray} \label{finishit}
\left(\int_{\Sigma_t \cap \{ r \ge r_{\max}\}} e_1\left[\widetilde{\psi}_\ell\right] \ r^2 dr \sin \theta d\theta d\phi \right)^{\frac{1}{2}} \geq || \widetilde{ \psi}_\ell ||_{H^1_{AdS}( \Sigma_t \cap \{ r \ge r_{\max}\})} \nonumber \\
\ge \frac{1}{2} ||  \widetilde{\psi}_\ell||_{H^1_{AdS}(\Sigma_0\cap \{ r \ge r_{\max}\}) }  \nonumber \\
\geq \frac{c}{2 \ell} \left[ || \Omega_i \widetilde{\psi}_\ell ||_{H^1_{AdS}(\Sigma_0 \cap \{ r \ge r_{\max}\})} + || \partial_{t^\star} \widetilde{\psi}_\ell ||_{H^1_{AdS}(\Sigma_0 \cap \{ r \ge r_{\max}\})}\right] \nonumber \\ 
\geq  \frac{c}{2 \ell} \left(\int_{\Sigma_0} e_2\left[\widetilde{\psi}_\ell\right] \ r^2 dr \sin \theta d\theta d\phi \right)^{\frac{1}{2}}. 
\end{eqnarray}
Here we have used -- in the step from the second to the third line -- that the data for $\widetilde{\psi}_\ell$ is frequency localized, which allows to exchange angular and time derivatives with powers of $\ell$ using the second item of Theorem \ref{th:main}, and radial derivatives by angular and time derivatives using the wave equation the $\widetilde{\psi}_\ell$ satisfy. From the third to the fourth line we exploited the fact that the data is localized in $r \ge r_{\max}$. The constant $c$ depends only on the parameters $M$, $l$, $a$ and $\alpha$.

Finally, setting $t_\ell=\frac{e^{C \ell}}{2C}$, we obtain from (\ref{finishit}) a family $(t_\ell, \psi_\ell)$ such that, for $\ell$ large enough,
$ Q \left[\widetilde{\psi}_\ell\right] \left(t_\ell\right) > C > 0$ holds for any $\ell$, which proves the Corollary.
\appendix
\section{The Improved boundedness statement} \label{app}
The  boundedness statement at the $H^2$-level proven in \cite{Holzegelwp, HolzegelAdS} is the estimate (\ref{bndhigher}) for the $\tilde{e}_2\left[\psi\right]$-based energies, cf.~Section \ref{sec:norms}. It is remarked in \cite{Holzegelwp} that stronger norms can be shown to be uniformly bounded using commutation by angular momentum operators leading to the statement (\ref{bndhigher}).  Since the latter statement has been used in this paper and also in \cite{gs:dpkads}, we provide here a sketch of the proof of this well-known (but absent from the literature) argument. We define the energies
\begin{align}
E_1\left[\psi\right]\left(t^\star\right) = \int_{\Sigma_{t^\star}} e_1 \left[\psi\right] \left(t^\star\right) \  r^2 \sin \theta dr d\theta d\phi
\end{align}
and with the obvious replacement, $E_2\left[\psi\right]\left(t^\star\right)$ and $\tilde{E}_2\left[\psi\right]\left(t^\star\right)$. Recall that uniform boudedness for the $\tilde{E}_2\left[\psi\right]$ energy is derived, in addition to known techniques near the horizon (cf.~the red-shift vector field), by commuting the Klein-Gordon equation with $\partial_t$ (which yields (\ref{bdness}) with $\psi$ replaced by $\partial_t \psi$) followed by elliptic estimates on spacelike slices, which control the $H^2_{AdS}$ norm. 

Let us sketch how to prove boundedness (\ref{bndhigher}) for the ${E}_2\left[\psi\right]$ energy. If we commute the Klein-Gordon equation with angular momentum operators we obtain
\begin{align} \label{commute}
\Box_g \left({\Omega}_i \psi \right) + \frac{\alpha}{l^2} \left({\Omega}_i \psi \right) \nonumber \\ =  2 {}^{({\Omega}_i)}\pi^{\mu \nu} \cdot \nabla_\mu \nabla_\nu \psi + \left[2\nabla^\alpha \left( {}^{({\Omega}_i)}\pi_{\alpha \mu} \right) - \nabla_\mu \left( {}^{({\Omega}_i)}\pi^{\alpha}_\alpha\right)  \right] \nabla^\mu \psi
\end{align}
with ${}^{({\Omega}_i)}\pi $ the (non-vanishing in Kerr!) deformation tensor of ${\Omega}_i$. The right hand side decays suitably in $r$ but not in $t$. More precisely, in view of the fact that there is no integrated decay estimate available, we cannot close the basic energy estimate on its own. Let us instead commute with localized angular momentum operators $\widetilde{\Omega}_i=\chi\left(r\right)\Omega_i$ where $\chi\left(r\right)$ is equal to $1$ for $r \geq 2R$ and equal to zero for $r\leq R$. Applying the energy estimate for the vectorfield $\partial_{t^\star}$, we can derive
\begin{align} \label{eest}
 ||\widetilde{\Omega}_i \psi_t ||_{H^{0,-2}_{AdS}(\Sigma_{t^\star} \cap \{r\geq 2R\})}+ ||\widetilde{\Omega}_i\psi ||_{H^1_{AdS}(\Sigma_{t^\star} \cap \{r\geq 2R\})} \nonumber \\
  \leq  C_{M,l,a,\alpha} \left[ ||\widetilde{\Omega}_i \psi_t ||_{H^{0,-2}_{AdS}(\Sigma_{t^\star_0})}+ ||\widetilde{\Omega}_i \psi ||_{H^1_{AdS}(\Sigma_{t^\star_0})}  \right] \nonumber \\
  +  \left(\tau_2-\tau_1\right)  \left[ \epsilon\sup_{\tau \in \left(\tau_1,\tau_2\right)} E_2\left[\psi\right]\left(\tau\right) + \epsilon \cdot \tilde{E}_2 \left[\psi\right] \left(0\right)\right] \, ,
\end{align}
where $\epsilon$ can be made small by choosing $R$ large. The last term arises from the spacetime error-term which decays strongly in $r$.

The idea is to combine this with an integrated decay estimate for the $\widetilde{\Omega}_i \psi$ which loses linearly in $\tau$. Recall that if $\Box \Psi + \frac{\alpha}{l^2}\Psi= f$, then we have the identity
\begin{align} \label{min}
 \nabla_a \Psi \nabla^a \psi - \frac{\alpha}{l^2}  \Psi^2 =  \nabla^\mu \left(\Psi \nabla_\mu \Psi \right) - g^{t^\star t^\star}\nabla_{t^\star} \Psi \nabla_{t^\star} \Psi - 2g^{t^\star b} \nabla_{t^\star} \Psi \nabla_b \Psi
  - f \Psi 
\end{align}
where $a,b$ run over $r,\theta, \phi$ only. When integrating this identity (with $\Psi$ replaced by $\widetilde{\Omega}_i \psi$ and $f$ being the error arising from the commutation in (\ref{commute})) with the usual spacetime volume we observe that 
\begin{itemize}
\item the left hand side is non-negative and controls all spatial derivatives after applying the standard Hardy inequality (cf.~(\ref{Hardy})), 

\item the second and third term on the right are essentially controlled by the $\tilde{E}_2\left[\psi\right]$ energy times the length of the time interval, $\left(\tau_2-\tau_1\right)$:
\[
g^{t^\star t^\star} \nabla_{t^\star}\left( \widetilde{\Omega}_i \psi\right)  \nabla_{t^\star} \left( \widetilde{\Omega}_i \psi \right) \sim \frac{1}{r^2} r^2 |\partial_{t^\star} \slashed{\nabla} \psi|_{\slashed{g}}^2 =  |\partial_{t^\star} \slashed{\nabla} \psi|_{\slashed{g}}^2 \leq e_1 \left[\partial_{t^\star} \psi\right]
\]
\[
g^{t^\star r} \nabla_{t^\star}\left( \widetilde{\Omega}_i \psi\right)  \nabla_{r} \left( \widetilde{\Omega}_i \psi \right) \sim \frac{1}{r^3} r^2 \left( |\partial_{t^\star} \slashed{\nabla} \psi|_{\slashed{g}}^2 + \| \partial_r \slashed{\nabla} \psi|_{\slashed{g}}^2 \right) \leq e_1 \left[\partial_{t^\star} \psi\right]
\]
\begin{align}
g^{t^\star \phi} \nabla_{t^\star}\left( \widetilde{\Omega}_i \psi\right)  \nabla_{\phi} \left( \widetilde{\Omega}_i \psi \right) \sim \frac{1}{r^2} r^2 \left( \frac{1}{\epsilon} |\partial_{t^\star} \slashed{\nabla} \psi|_{\slashed{g}}^2 + \epsilon \| \partial_{\phi} \slashed{\nabla} \psi|_{\slashed{g}}^2 \right) \nonumber \\
\leq C_\epsilon \cdot e_1 \left[\partial_{t^\star} \psi\right] + \epsilon \cdot {e}_2 \left[\psi\right]
\end{align}
and $g^{t^\star \theta}=0$. The fact that we need to borrow an $\epsilon$ of ${e}_2\left[\psi\right]$ is due to the fact that the $t^\star \phi$ coordinates are not optimal near infinity. The cross-term $g^{t^\star \phi}$ term would decay much stronger in coordinates adapted to the asymptotically AdS end, which would allow us to estimate all terms by the weaker energy $\tilde{E}_2\left[\psi\right]$. 

\item the first term on the right hand side is a boundary term, which can be estimated 
\begin{align}
\Big| \int_{\mathcal{D}\left(\tau_1,\tau_2\right)}  \nabla^\mu \left(\widetilde{\Omega}_i \psi \nabla_\mu \left(\widetilde{\Omega}_i \psi\right) \right) \Big| \leq \sup_{t^\star} \int_{\Sigma_{t^\star}} \tilde{e}_2\left[\psi\right] r^2 \sin \theta dr d\theta d\phi \, .
\end{align}

\item the last term in (\ref{min}) is controlled as previously by the last line in the energy estimate (\ref{eest}).
\end{itemize}
It follows that integrating (\ref{min}) furnishes the estimate
\begin{align} \label{ied}
\int_{\mathcal{D}\left(\tau_1,\tau_2\right) \cap \{r\geq 2R\}} r^2 \sin \theta  dt^\star dr d\theta d\phi \Big[ r^2 |\partial_r \Omega_i \psi|^2 +  | \slashed{\nabla} \Omega_i \psi|^2 \nonumber \\
\leq \max\left(1, \left(\tau_2-\tau_1\right)\right)  \left[ \epsilon\sup_{\tau \in \left(\tau_1,\tau_2\right)} {E}_2\left[\psi\right] \left(\tau\right) + C \cdot \tilde{E}_2\left[\psi\right] \left(0\right)\right] \, .
\end{align}
Now note that
\begin{align} \label{locen}
\tilde{E}_2\left[\psi\right]\left(t^\star\right) +  ||\Omega_i \psi_t ||^2_{H^{0,-2}_{AdS}(\Sigma_{t^\star} \cap \{r\leq 2R\})}+ ||\Omega_i \psi ||^2_{H^1_{AdS}(\Sigma_{t^\star} \cap \{r\leq 2R\})} \nonumber \\  
\leq C_R \cdot \tilde{E}_2\left[\psi\right]\left(t^\star\right)  
 \leq C_R \cdot \tilde{E}_2\left[\psi\right]\left(0\right)
 \end{align}
follows right from the boundedness statement for the $\tilde{E}_2\left[\psi\right]$ energy and estimating the weights away from infinity. We can integrate (\ref{locen}) in time and
add it to (\ref{ied}) which yields (first without the boxed terms)
\begin{align} \label{boot}
\boxed{ {E}_2 \left[\psi\right] \left(\tau_2 \right) }+ \int_{\tau_1}^{\tau_2}{E}_2  \left[\psi\right] \left(\tau\right) d\tau
 \leq \boxed{C_{M,l,a,\alpha} \cdot {E}_2  \left[\psi\right]  \left(\tau_1 \right)}  \nonumber \\ + \max\left(1, \tau_2-\tau_1\right)  \left[ \epsilon\sup_{\tau \in \left(\tau_1,\tau_2\right)} {E}_2 \left[\psi\right] \left(\tau\right) + C_\epsilon \cdot \tilde{E}_2 \left[\psi\right] \left(0\right)\right] \, .
\end{align}
The estimate also holds with the boxed terms included, as follows from adding (\ref{eest}) and (\ref{locen}). We claim that (\ref{boot}) implies ${E}_2 \left[\psi\right]\left(t^\star\right) \lesssim{E}_2\left[\psi\right]\left(0\right)$ provided $\epsilon$ is sufficiently small depending only on the parameters (the constant $C_{M,l,a,\alpha}$) and leave the verification to the reader.

\begin{remark}
An easier proof is available if one is willing to go to $H^3_{AdS}$. The Carter operator
\begin{align}
Q \psi = \Delta_{\mathbb{S}^2} \psi - \partial_\phi^2 \psi + \left(a^2 \sin^2 \theta\right) \partial_t^2 \psi
\end{align}
commutes with the wave operator. Since $\partial_\phi^2$ and $\partial_t^2$ trivially commute, we have
\begin{align}
E_1\left[\Delta_{\mathbb{S}^2} \psi \right] \left(t^\star\right) &\lesssim  E_1\left[Q \psi \right]  \left(t^\star\right) + E_1\left[\partial_\phi^2 \psi \right]  \left(t^\star\right) +E_1\left[\partial_t^2\psi \right]  \left(t^\star\right) \nonumber \\
&\lesssim  E_1\left[Q \psi \right]  \left(0\right) + E_1\left[\partial_\phi^2 \psi \right]  \left(0\right) +E_1\left[\partial_t^2\psi \right]  \left(0\right)
\end{align}
and we can control all derivatives on $S^2$ from controlling the Laplacian via elliptic estimates. This yields the desired gain, albeit at the level of three derivatives. This is analogous to commuting with angular momentum operators twice.
\end{remark}

\bibliographystyle{hacm}
\bibliography{refs}
\end{document}